\RequirePackage{fix-cm}
\documentclass[smallextended]{svjour3}       
\def\makeheadbox{\relax}

\smartqed 
\usepackage{graphicx}
\usepackage{amsmath,amsopn,amssymb}

\usepackage{amsthm}
\usepackage{verbatim}
\usepackage{tikz}
\usetikzlibrary{shapes.geometric}
\usepackage{bm}
\usetikzlibrary{calc}
\usetikzlibrary{positioning}
\usepackage{algpseudocode}

\def\change#1{{#1}}

\newcommand{\DelayClin}{{\sf DelayC_{lin}}}
\newcommand{\DelayLin}{{\sf DelayLin}}
\newcommand{\encode}[1]{||#1||}
\newcommand{\qDecide}[2]{\text{\sc Decide}_{#2}\langle#1\rangle}
\newcommand{\pEnum}[1]{\text{\sc Enum}\langle#1\rangle}
\newcommand{\qEnum}[2]{\text{\sc Enum}_{#2}\langle#1\rangle}

\newcommand{\idom}{{\it dom}}
\newcommand{\arity}{\operatorname{arity}}
\newcommand{\var}{{\it var}}
\newcommand{\free}{\operatorname{free}}
\newcommand{\atoms}{\operatorname{atoms}}
\newcommand{\dltplus}{\Delta_{Q^+}}
\newcommand{\tetra}[1]{\text{\sc Tetra}{\normalfont (}#1{\normalfont )}}
\newcommand{\tetpm}[1]{{\sf Tet_{pm}}{\normalfont (}#1{\normalfont )}}
\newcommand{\rel}[1]{#1_{\operatorname{base}}}
\newcommand{\dis}[1]{\operatorname{dis}(#1)}
\newcommand{\SJF}{\text{\sc{SJF}}}

\newcommand{\calC}{{\mathcal C}}
\newcommand{\calD}{{\mathcal D}}
\newcommand{\calG}{{\mathcal G}}
\newcommand{\calH}{{\mathcal H}}
\newcommand{\calO}{{\mathcal O}}
\newcommand{\calR}{{\mathcal R}}
\newcommand{\calS}{{\mathcal S}}

\smartqed

\newcommand{\convexpath}[2]{
	[   
	create hullnodes/.code={
		\global\edef\namelist{#1}
		\foreach [count=\counter] \nodename in \namelist {
			\global\edef\numberofnodes{\counter}
			\node at (\nodename) [draw=none,name=hullnode\counter] {};
		}
		\node at (hullnode\numberofnodes) [name=hullnode0,draw=none] {};
		\pgfmathtruncatemacro\lastnumber{\numberofnodes+1}
		\node at (hullnode1) [name=hullnode\lastnumber,draw=none] {};
	},
	create hullnodes
	]
	($(hullnode1)!#2!-90:(hullnode0)$)
	\foreach [
	evaluate=\currentnode as \previousnode using \currentnode-1,
	evaluate=\currentnode as \nextnode using \currentnode+1
	] \currentnode in {1,...,\numberofnodes} {
		-- ($(hullnode\currentnode)!#2!-90:(hullnode\previousnode)$)
		let \p1 = ($(hullnode\currentnode)!#2!-90:(hullnode\previousnode) - (hullnode\currentnode)$),
		\n1 = {atan2(\y1,\x1)},
		\p2 = ($(hullnode\currentnode)!#2!90:(hullnode\nextnode) - (hullnode\currentnode)$),
		\n2 = {atan2(\y2,\x2)},
		\n{delta} = {-Mod(\n1-\n2,360)}
		in 
		{arc [start angle=\n1, delta angle=\n{delta}, radius=#2]}
	}
	-- cycle
}

\begin{document}

\title{Enumeration Complexity of Conjunctive Queries with Functional Dependencies\thanks{This
		work was supported
		by the German-Israeli Foundation for Scientific Research and Development (GIF), Grant no.~I-2436-407.6/2016
		and
		by the Austrian Science Fund 
		(FWF): W1255-N23, P25207-N23, P25518-N23.}
}

\titlerunning{Enumeration Complexity of CQs with FDs}

\author{Nofar Carmeli         \and
        Markus Kr\"oll
}

\authorrunning{N. Carmeli, M. Kr\"oll}

\institute{N. Carmeli \at
              Technion, Haifa, Israel, \\
              \email{snofca@cs.technion.ac.il}
           \and
           M. Kr\"oll \at
              TU Wien, Vienna, Austria,\\
              \email{kroell@dbai.tuwien.ac.at}
}

\date{}

\maketitle

\begin{abstract}
We study the complexity of enumerating the answers of Conjunctive Queries (CQs) in the presence of Functional Dependencies (FDs). Our focus is on the ability to list output tuples with a constant delay in between, following a linear-time preprocessing. A known dichotomy classifies the acyclic self-join free CQs into those that admit such enumeration, and those that do not. However, this classification no longer holds in the common case where the database exhibits dependencies among attributes. That is, some queries that are classified as hard are in fact tractable if dependencies are accounted for. We establish a generalization of the dichotomy to accommodate FDs; hence, our classification determines which combination of a CQ and a set of FDs admits constant-delay enumeration with a linear-time preprocessing.

In addition, we generalize a hardness result for cyclic CQs to accommodate unary FDs, and further conclusions of our development include a dichotomy for enumeration with linear delay. Finally, we show that all our results apply also for CQs with disequalities and in the presence of cardinality dependencies that generalize FDs.
\keywords{Enumeration \and Complexity \and Conjunctive queries \and Functional dependencies}
\end{abstract}

\section{Introduction}
\label{intro}
When evaluating a non-boolean Conjunctive Query (CQ) over a database, the number of results can be huge.
Since this number may be larger than the size of the database itself, we need to use specific measures of 
enumeration complexity  to describe the hardness of such a problem. In this perspective, the best we 
can hope for is to constantly output results, in such a way that the delay between them is unaffected by the size
of the database instance. For this to be possible, we need to allow a precomputation phase, as linear time preprocessing is necessary to read the input instance before deciding the existence of a first result.

A known dichotomy determines when the answers to self-join free acyclic CQs can be enumerated with constant delay after linear time preprocessing~\cite{bdg:dichotomy}. This class of enumeration problems, denoted by $\DelayClin$, can be regarded as the most efficient class of nontrivial enumeration problems 
and therefore current work on query enumeration has focused on this class \cite{Durand:2014:EAF:2594538.2594539,DBLP:conf/icdt/SegoufinV17,DBLP:conf/pods/BerkholzKS17}.
Bagan et al.~\cite{bdg:dichotomy} show that a subclass of acyclic queries, called \emph{free-connex},
are exactly those that are enumerable in $\DelayClin$, under the common assumption that boolean
matrix multiplication cannot be solved in quadratic time. 
An acyclic query is called free-connex if the query remains acyclic when treating the head of the query as 
an additional atom.
This and all other results in this paper hold under the RAM model~\cite{DBLP:journals/amai/Grandjean96}.

The above mentioned dichotomy only holds when applied to databases with no additional assumptions, but oftentimes this is not the case. In practice, there is usually a connection between different attributes, and \emph{Functional Dependencies} (FDs) and \emph{Cardinality Dependencies} (CDs) are widely used to model situations where some attributes imply others. As the following example shows, these constraints
also have an immediate effect on the complexity of enumerating answers for queries over such a schema.
\begin{example}\label{example:fd-dif1}
	For a list of actors and the production companies they work with, consider the query
\begin{equation*}
	 Q(actor,production)\leftarrow \text{Cast}(movie,actor), \text{Release}(movie,production).
\end{equation*}	
	At first glance, it appears as though  this query is not in $\DelayClin$, as it is acyclic but not free-connex. Nevertheless, if we take the fact that a movie has only one production company into account,
	we have the FD $\text{Release}:1\rightarrow 2$, and the enumeration problem becomes easy: 
	we only need to iterate over all tuples of $\text{Cast}$ and replace the $movie$ value with the single $production$ value that the relation $\text{Release}$ assigns to it. This can be done in linear time by first
	building a lookup table from $\text{Release}$ in linear time. 
	\qed
\end{example}

Example~\ref{example:fd-dif1} shows that
the dichotomy by Bagan et al.~\cite{bdg:dichotomy} does not hold in the presence of FDs. 
In fact, we believe that dependencies between attributes are so common in real life that 
ignoring them in such dichotomies can lead to missing a significant portion of the tractable cases.
Therefore, to get a realistic picture of the enumeration complexity of CQs, we have to take dependencies into account.
The goal of this work is to generalize the dichotomy to fully accommodate FDs.

Towards this goal,
we introduce an extension of a query $Q$ according to the FDs. The extension is called the FD-extended query, and denoted $Q^+$.
In this extension, each atom, as well as the head of the query, contains all variables that can be implied by its variables according to the FDs.
This way, instead of classifying every combination of CQ and FDs directly, we encode the dependencies within the extended query, and use the classification of $Q^+$ to gain insight regarding $Q$. 
This approach draws inspiration from the proof of a dichotomy in the complexity of \emph{deletion propagation}, in the presence of FDs~\cite{Kimelfeld:2012:DCD}.
However, the problem and consequently the proof techniques are fundamentally different.

The FD-extension is defined in such a way that if $Q$ is satisfied by an assignment, 
then the same assignment also satisfies the extension $Q^+$, as the underlying instance is bound
by the FDs. 
In fact, we can show that enumerating the solutions of $Q$ under FDs can be reduced to enumerating the solutions of $Q^+$. Therefore, tractability of $Q^+$ ensures that $Q$ can be efficiently solved as well. By using the positive result in the known dichotomy, $Q^+$ is tractable w.r.t enumeration if it is free-connex. 
Moreover, it can be shown that the structural restrictions of acyclicity and free-connex are closed
under FD-extensions. Hence, the class of
all queries $Q$ such that $Q^+$ is free-connex is a proper extension of the class of free-connex queries.
We denote the classes of queries $Q$ such that $Q^+$ is acyclic or free-connex as FD-acyclic or FD-free-connex, respectively.

To reach a dichotomy, we now need to answer the following question: Is it possible that $Q$ can be enumerated efficiently even if it is not FD-free-connex?
To show that an enumeration problem is not within a given class, enumeration complexity has few tools to offer. One such tool is a notion of completeness
for enumeration problems~\cite{DBLP:conf/lata/CreignouKPSV17}. However, this notion focuses on problems with a complexity corresponding to higher classes of the polynomial hierarchy. So in order to deal with this problem, Bagan et al.~\cite{bdg:dichotomy} reduced the matrix multiplication problem to enumerating the answers to any query that is acyclic but not free-connex.
This reduction fails, however, when dependencies are imposed on the data, as the constructed database instance does not necessarily satisfy the underlying dependencies.

As it turns out, however, the structure of the FD-extended query $Q^+$ allows us to extend this reduction
to our setting.
By carefully expanding the reduced instance such that on the one hand, the dependencies
hold and on the other hand, the reduction can still be performed within linear time, we establish a dichotomy.
That is, we show that the tractability of enumerating the answers of a self-join free 
query $Q$ in the presence of FDs is exactly characterized by the structure of $Q^+$: Given an FD-acyclic query $Q$,
we can enumerate the answers to $Q$ within the class $\DelayClin$ iff $Q$ is FD-free-connex.

The resulting extended dichotomy, as well as the original one, brings insight to the case of acyclic queries.
Concerning unrestricted CQs, providing even a first solution of a query in linear time is impossible in general.
This is due to the fact that the parameterized complexity of answering boolean CQs,
taking the query size as the parameter,
is $\mathsf{W}[1]$-hard \cite{DBLP:journals/jcss/PapadimitriouY99}. 
This does not imply, however,
that there are no cyclic queries with the corresponding enumeration problems in $\DelayClin$. 
The fact that no such queries exist requires an additional proof, which was presented by Brault-Baron~\cite{bb:thesis}. This result holds under the assumption that a generalized version of the triangle finding problem is not
solvable within linear time \cite{DBLP:conf/focs/WilliamsW10}.
As before, this proof does no longer apply in the presence of FDs. Moreover, it is possible for $Q$ to be cyclic and $Q^+$ acyclic. In fact, $Q^+$ may even be free-connex, and therefore tractable in $\DelayClin$.
We show that, under the same assumptions used by Brault-Baron~\cite{bb:thesis}, the evaluation problem for a 
self-join free CQ in the presence of unary FDs where $Q^+$ is cyclic cannot be solved in linear time.
As linear time preprocessing is not enough to achieve the first result,
a consequence is that enumeration within $\DelayClin$ is impossible in that case.
This covers all types of self-join free CQs and shows a full dichotomy for the case of unary FDs. Moreover, we show how our results can be easily used to yield additional results, such as a dichotomy for the evaluation of CQs with linear delay.

The results we present here are not limited to CQs and FDs.
CQs with disequalities are an extension of CQs, allowing to restrict the satisfying assignments and demand that some pairs of variables map to different values. We prove that our results apply to the more general query class: CQs with disequalities.

Another way of generalizing our results is not to extend the class of queries, but the class of dependencies.
Cardinality Dependencies (CDs) \cite{DBLP:journals/pvldb/CaoFWY14,DBLP:conf/icalp/ArapinisFG16}
are a generalization of FDs, denoted $(R_i:A\rightarrow B,c)$. Here, the right-hand side does not have to be unique for every assignment to the left-hand side, but there can be at most $c$ different values to the variables of $B$ for every value of the variables of $A$. FDs are in fact a special case of CDs where $c=1$. Constraints of that form appear naturally in many applications. For example: a movie has only a handful of directors, there are at most 200 countries, and a person is typically
limited to at most 5000 friends in (some) social networks.
We show that all results described in this paper also apply to CDs.

\paragraph{Contributions.}
Our main contributions are as follows.
\begin{itemize}
	\item We extend the class of queries that are known to be in $\DelayClin$ by incorporating FDs.
	This extension is the class of FD-free-connex CQs.
	\item We establish a dichotomy for the enumeration complexity of self-join free
	FD-acyclic CQs. Consequently, we get a dichotomy for self-join free acyclic CQs in the presence of FDs.
	\item We show a lower bound for FD-cyclic CQs. In particular, we get a dichotomy for all self-join free CQs in the presence of unary FDs with respect to $\DelayClin$, and a similar dichotomy regarding enumeration with linear delay.
	\item We show the extension of all our results to CDs and CQs with disequalities.
\end{itemize}

A short version of this article appeared in the 21st International Conference on Database Theory (ICDT18)~\cite{icdtVersion}. This article has the following additions.
First, it includes all proofs omitted from the short version. Second, it includes a new section, Section~\ref{sec:disequalities}, proving that all results in this paper hold also for CQs with disequalities. In addition, we add more thorough discussions both on the computational model and on the complexity assumption used in Section~\ref{sec:cyclic}.

This work is organized as follows: in Section~\ref{sec:Preliminaries}
we provide definitions and state previous results that we will use.
Section~\ref{sec:extending} introduces the notion of FD-extended queries
and establishes the equivalence between a query and its FD-extension.
The dichotomy for acyclic CQs is shown in Section~\ref{sec:Positive}, while Section~\ref{sec:cyclic} shows a lower bound for cyclic queries under unary FDs.
Section~\ref{section:cardinalitydep} shows that all results from the previous sections extend to CDs, and Section~\ref{sec:disequalities} does the same for CQs with disequalities. Concluding remarks are given in Section~\ref{section:conclusion}.

\section{Preliminaries} \label{sec:Preliminaries}
In this section we provide preliminary definitions as well as state results that we will use throughout this paper.

\subsubsection*{Schemas and Functional Dependencies}
A {\em schema} $\calS$ is a pair $(\calR,\Delta)$ where $\calR$ is a finite set $\{R_1,\ldots,R_n\}$ of {\em relational symbols} and $\Delta$ is a set of {\em Functional Dependencies} (FDs).
We denote the {\em arity} of a relational symbol $R_i$
as $\arity(R_i)$. An FD $\delta\in\Delta$ has the form
$R_i\colon A\rightarrow B$, where $R_i\in\calR$ and $A,B$ are non-empty with $A,B\subseteq\{1,\ldots,\arity(R_i)\}$.

Let $\idom$ be a finite set of constants. A database $I$ over schema $\calS$ is called an 
{\em instance} of $\calS$, and it consists of a finite relation $R_i^I\subseteq\idom^{\arity(R_i)}$ for every relational symbol $R_i\in\calR$, such that all FDs in $\Delta$ are {\em satisfied}.
An FD $\delta=R_i\colon A\rightarrow B$ is said to be satisfied if, for all tuples $u,v\in R_i^{I}$ that are equal on the indices of $A$, $u$ and $v$ are equal on the indices of $B$. 
In what follows, we will often assume w.l.o.g. that all FDs are of the form
$R_i\colon A\rightarrow b$, where $b\in\{1,\ldots,\arity(R_i)\}$, as we can replace an FD of the form $R_i\colon A\rightarrow B$ where $|B|>1$ by the set of FDs $\{R_i\colon A\rightarrow b\mid b\in{B}\}$.
If $|A|=1$, we say that $\delta$ is a {\em unary} FD.

\subsubsection*{Conjunctive Queries} 
Let $\var$ be a set of variables disjoint from $\idom$.
A {\em Conjunctive Query} (CQ) over schema $\calS=(\calR,\Delta)$ is an expression of the form 
$Q(\vec{p}) \leftarrow R_1(\vec{v}_1), \dots, R_m(\vec{v}_m)$,
 where $R_1,\ldots,R_m$ are relational symbols of $\calR$, the tuples $\vec{p}, \vec{v}_1,\ldots, \vec{v}_m$ hold variables, and every variable in $\vec{p}$ appears in at least one of $\vec{v}_1,\ldots, \vec{v}_m$. 
We often denote this query as
$Q(\vec{p})$ or even $Q$. Define the variables of $Q$ as
$\var(Q)=\bigcup_{i=1}^{m} \vec{v}_i$,
and define the {\em free variables}
of $Q$ as $\free(Q)=\vec{p}$.
We call $Q(\vec{p})$ the head of $Q$, and the atomic formulas $R_i(\vec{v}_i)$ are called {\em atoms}.
We further use $\atoms(Q)$ to denote the set of atoms of Q.
A CQ is said to contain {\em self-joins} if some relation symbol appears in more than one atom.

For the {\em evaluation} $Q(I)$ of a CQ $Q$ with free variables $\vec{p}$ over a database $I$, 
we define $Q(I)$ to be the set of all \emph{mappings}
$\mu|_{\vec{p}}$ such that $\mu$ is a homomorphism from 
$R_1(\vec{v}_1), \dots , R_m(\vec{v}_m)$ into $I$, where $\mu|_{\vec{p}}$
denotes the restriction (or projection) of $\mu$ to the variables
$\vec{p}$.
The problem $\qDecide{Q}{\Delta}$ is,
given a database instance $I$ over a schema $(\calR,\Delta)$,
determining whether such a mapping exists.

Let $R(\vec{v})$ be an atom of a CQ. We say that a tuple $t\in R^I$ assigns a variable $x$ with the value $c$ if for every index $i$ such that $\vec{v}[i]=x$ we have that $t[i]=c$.
We say that a tuple $t_a\in R_a^I$ agrees with a tuple $t_b\in R_b^I$ on the value of a variable $x$ if for every vectors $\vec{v_a}, \vec{v_b}$ such that $R_a(\vec{v_a})$ and $R_b(\vec{v_b})$ are atoms of the CQ, and for every pair of indices $i_a,i_b$ such that $\vec{v_a}[i_a]=\vec{v_b}[i_b]=x$, we have that $t_a[i_a]=t_b[i_b]$.
Given a query $Q$ over a schema $\calS=(\calR,\Delta)$, we often identify an FD $\delta\in\Delta$ as a mapping between variables. That is, if $\delta$ has the form
$R_i:A\rightarrow b$ for $A=\{a_1,\ldots,a_{|A|}\}$,
we sometimes denote it by $R_i:\{\vec{v}_i[a_1],\ldots,\vec{v}_i[a_{|A|}]\}\rightarrow\vec{v}_i[b]$, where $\vec{u}[k]$ is the $k$-th variable of $\vec{u}$. To distinguish between these two representations, we usually
denote subsets of integers by $A,B,C,\ldots$, integers by $a,b,c,\ldots$, and
variables by letters from the end of the alphabet.

\subsubsection*{Hypergraphs}

A {\em hypergraph} $\calH=(V,E)$ is a pair consisting of a set $V$ of {\em vertices}, and a set $E$ of non-empty subsets of $V$ called {\em hyperedges} (sometimes {\em edges}).
A \emph{join tree} of a hypergraph $\calH=(V,E)$ is a tree $T$ where the nodes are the hyperedges of $\calH$, and the {\em running intersection} property holds, namely: for all $u \in V$ the set $\{e \in E \mid u \in e\}$ forms a connected subtree in $T$. 
A hypergraph $\calH$ is said to be {\em acyclic} if there exists a join tree for $\calH$.
Two vertices in a hypergraph are said to be {\em neighbors} if they appear in the same edge.
A {\em clique} of a hypergraph is a set of vertices, which are pairwise neighbors in $\calH$.
A hypergraph $\calH$ is said to be {\em conformal} if every clique of $\calH$ is contained in some edge of $\calH$.
A {\em chordless cycle} of $\calH$ is a tuple $(x_1,\ldots,x_n)$, $n\geq3$, such that the set of neighboring pairs of variables of $\{x_1,\ldots,x_n\}$ is exactly $\{\{x_i,x_{i+1}\}\mid 1\leq i\leq n-1\}\cup\{\{x_n,x_1\}\}$. It is well known (see~\cite{DBLP:journals/jacm/BeeriFMY83})
that a hypergraph is acyclic iff it is conformal and contains no chordless cycles.

A {\em pseudo-minor} of a hypergraph $\calH=(V,E)$ is a hypergraph obtained from $\calH$ by a finite series of the following operations:
\begin{enumerate}
\item {\em vertex removal}: removing a vertex from $V$ and from all edges in $E$ that contain it.
\item {\em edge removal}: removing an edge $e$ from $E$ provided that some other $e'\in E$  contains it.
\item {\em edge contraction}: replacing all occurrences of a vertex $v$ (within every edge) with a vertex $u$, provided that $u$ and $v$ are neighbors.
\end{enumerate}

\subsubsection*{Classes of CQs}
To a CQ $Q$ we associate a hypergraph $\calH(Q)=(V,E)$ where the vertices $V$ are the variables of $Q$ and every hyperedge $E$ is a set of variables occurring in a single atom of $Q$, that is $E=\{\{v_1,\ldots,v_n\}\} \mid R_i(v_1,\ldots,v_n)\in\atoms(Q)\}$. With a slight abuse of notation, we also identify atoms
of $Q$ with edges of $\calH(Q)$.
A CQ $Q$ is said to be {\em acyclic} if $\calH(Q)$ is acyclic, and it is said to be {\em free-connex} if both $Q$ and $(V, E\cup\{\free(Q)\})$ are acyclic.

A {\em head-path} for a CQ $Q$ is a sequence of variables $(x,z_1,\ldots,z_k,y)$ with $k\geq 1$, such that:
\begin{enumerate}
\item $\{x,y\}\subseteq \free(Q)$
\item  $\{z_1,\ldots,z_k\}\subseteq V\setminus\free(Q)$
\item It is a {\em chordless path} in $\calH(Q)$, that is, two succeeding variables appear together in some atom, and no two non-succeeding variables appear together in an atom.
\end{enumerate}

For a head-path $P=(z_0,z_1,\ldots,z_k,z_{k+1})$, the \emph{length} of $P$ is defined as $k+1$.
Bagan et al.~\cite{bdg:dichotomy} showed that
an acyclic CQ has a head-path iff it is not free-connex.

\subsubsection*{Computational Model}

In this work, we only consider finite structures. We thus always assume that an input $w$ is a string
over a (fixed) finite alphabet, and we denote by $\encode{w}$ the size of $w$, which is the length
of the string. 
Using data complexity for most of our problems, the input is measured only by the size of
the database instance $I$.
Let $I$ be a database over a schema $\calS=(\calR,\Delta)$.
Flum et al. describe a reasonable encoding $\encode{I}$ of the database as a
string over a finite alphabet~\cite{DBLP:journals/jacm/FlumFG02} with
$\encode{R^I} = \arity(R)|R^{I}|$ and 
$\encode{I}= 1 + |\idom| + |\calR| + \sum_{R\in\calR}{\encode{R^I}}$.
When we say linear time, we mean that the number of operations is $\calO(\encode{I})$.

In this paper we adopt the \emph{Random Access Machine} (RAM) model with uniform cost measure.
For an input of size $n$, every register is of length $\calO(\log(n))$. Operations such as addition of
the values of two registers or concatenation can be performed in constant time. In contrast to the
Turing model of computation, the RAM model with uniform cost measure can retrieve the content of
any register via its unique address in constant time. This enables the construction of large lookup tables
that can be queried in constant time. 

We use a variant of the RAM model named DRAM~\cite{DBLP:journals/amai/Grandjean96}, where the
values stored in registers are at most $n^c$ for some fixed integer $c$. As a consequence, the amount
of available memory is polynomial in $n$. Some previous work mentioned in this paper use a more
restrictive variant of the model, called a DLINRAM~\cite{DBLP:journals/amai/Grandjean96}, where the
values are at most $cn/log(n)$, and therefore the available memory is linear in the size of the input.
All results presented in this article except for the positive results for cardinality dependencies (that
use Lemma~\ref{lemma:posCD}) also apply with the DLINRAM model. The lower bounds shown on the
DRAM model imply the same lower bounds for the more restrictive DLINRAM model. The positive results
in the presence of FDs use the reduction described in Theorem~\ref{thm:qpEquiv}. There, the set of keys
is a projected subset of the input relations, so the memory used is bound by the size of the input.

Grandjean proved that sorting strings in the DLINRAM model can be done in time $\calO(n / \log n)$, 
where $n$ is the size of the input containing strings encoded in some fixed alphabet and separated
by some special symbol~\cite{DBLP:journals/amai/Grandjean96}. We can use this method to sort relations.
To construct the input to the sorting algorithm, we first translate the values from $\idom$ to a possibly
smaller domain $\idom_R$, containing only the values that appear in $R^I$. Note that
$|\idom_R|\leq\encode{R^I}$. Then, we translate these values to binary
(since we are required to use a fixed alphabet), where each value takes $\log(\idom_R)$ bits.
The size of the input to the sorting problem is
$n=\encode{R^I}\cdot\log(|\idom_R|)+(|R^{I}|-1)$.
Therefore, we can sort the tuples of a relation in time $\calO(n / \log n)=\calO(\encode{R^I})$.
That is, it is possible to sort relations within linear time~\cite{DBLP:conf/icdt/SegoufinV17}.

\subsubsection*{Enumeration Complexity}
Given a finite alphabet $\Sigma$ and binary relation $R\subseteq\Sigma^*\times\Sigma^*$,
we denote by $\pEnum{R}$ the {\em enumeration problem} of given an instance $x\in\Sigma^*$, to output all $y\in\Sigma^*$ such that $(x,y)\in R$. 
For a class $\calC$ of enumeration problems, we say that $\pEnum{R}\in\calC$, if there is a RAM that, on input $x\in\Sigma^*$, outputs all $y\in\Sigma^*$ with
$(x,y)\in R$ without repetition such that the first output is computed in time $p(|x|)$ and the delay
between any two consecutive outputs after the first is $d(|x|)$, where:
\begin{itemize}
	\item For $\calC=\DelayClin$, we have $p(|x|)\in O(|x|)$ and $d(|x|)\in O(1)$.
	\item For $\calC=\DelayLin$, we have $p(|x|),d(|x|)\in O(|x|)$.
\end{itemize}
Let $\pEnum{R_1}$ and $\pEnum{R_2}$ be enumeration problems. There is an {\em exact reduction} from $\pEnum{R_1}$ to $\pEnum{R_2}$, denoted $\pEnum{R_1}\leq_e\pEnum{R_2}$, if there are
mappings $\sigma$ and $\tau$ such that the following conditions hold.
\begin{itemize}
	\item $\{\tau(y)\mid y\in\Sigma^*\text{ with }(\sigma(x),y)\in R_2\} =
	\{y'\in\Sigma^*\mid (x,y')\in R_1\}$ in multiset notation.
	\item for every $x\in\Sigma^*$,
	the mapping $\sigma(x)$ is computable in $\calO(|x|)$;
	\item for every $y\in\Sigma^*$ with $(\sigma(x),y)\in R_2$, $\tau(y)$ is computable in constant time;
\end{itemize}
Intuitively, $\sigma$ is used to map instances of $\pEnum{R_1}$ to instances of $\pEnum{R_2}$,
and $\tau$ is used to map solutions to $\pEnum{R_2}$ to solutions of $\pEnum{R_1}$.
The notation $\pEnum{R_1}\equiv_e\pEnum{R_2}$ means that $\pEnum{R_1}\leq_e\pEnum{R_2}$ and $\pEnum{R_2}\leq_e\pEnum{R_1}$.
An enumeration class $\calC$ is said to be \emph{closed under exact reduction} if for
every $\pEnum{R_1}$ and $\pEnum{R_2}$ with $\pEnum{R_2}\in\calC$ and $\pEnum{R_1}\leq_e\pEnum{R_2}$, we have that $\pEnum{R_1}\in\calC$.
Bagan et al.~\cite{bdg:dichotomy} proved that $\DelayClin$ is closed under exact reduction.
The same proof holds for any meaningful enumeration complexity class that guarantees generating all unique answers with at least linear preprocessing time and at least constant delay between answers.

\subsubsection*{Enumerating Answers to CQs}
For a CQ $Q$ over a schema $\calS=(\calR,\Delta)$, we denote by $\qEnum{Q}{\Delta}$ the enumeration problem $\pEnum{R}$, where $R$ is the binary relation between instances $I$ over $\calS$ and sets of mappings $Q(I)$. We consider the size of the query as well as the size of the schema to be fixed. 
Bagan et al.~\cite{bdg:dichotomy} showed that a self-join free acyclic CQ is in $\DelayClin$ iff it is free-connex:

\begin{theorem}[\cite{bdg:dichotomy}]\label{theorem:originalDichotomy}
	Let $Q$ be an acyclic CQ over a schema $\calS=(\calR,\emptyset)$.
	\begin{enumerate}
		\item If $Q$ is free-connex, then $\qEnum{Q}{\emptyset}\in\DelayClin$.
		\item If $Q$ is not free-connex \change{and it is self-join free},  then $\qEnum{Q}{\emptyset}\not\in\DelayClin$, assuming the product of
		two $n \times n$ boolean matrices cannot be computed in time $\calO(n^2)$.
	\end{enumerate}
\end{theorem}

\section{FD-Extended CQs}\label{sec:extending}

In this section, we formally define the extended query $Q^+$.
Intuitively, the extension is based on treating the FDs as
dependencies between variables.
Assume that a variable $x$ implies a variable $y$ via an FD $R: x\rightarrow y$.
Given an assignment for $x$ in some answer to the query, we know the unique assignment of $y$ according to $R$. Therefore, we can add the $y$ value to every tuple that has an $x$ value.
Note however that when we extend the relations with these known values, the data still conforms to the FDs. Thus, the process of extension does not remove the FDs. Instead, it ensures that the extended query holds in its structure the added value that the FDs provide. 
In Definition~\ref{definition:fdextended} we formalize the extension of the schema and the query, and in Theorem~\ref{thm:qpEquiv} we formalize the extension of the instance.

We also discuss the relationship between $Q$ and $Q^+$: their equivalence w.r.t.~enumeration
and the possible structural differences between them. 
As a result, we obtain that if $Q^+$ is in a class of queries that allows for tractable enumeration, 
then $Q$ is tractable as well.

The {\em extension} of an atom $R(\vec{v})$ according to an
FD $S\colon A\rightarrow b$ and an atom $S(\vec{u})$
is possible if $\vec{u}[A]\subseteq\vec{v}$
but $\vec{u}[b]\notin\vec{v}$. In this case, $\vec{u}[b]$ is added to the variables of $R$.
The {\em FD-extension} of a query is defined by iteratively extending all
atoms as well as the head according to every possible dependency in the schema, until a fixpoint is reached.
The schema extends accordingly: the arities of the relations increase as their corresponding atoms extend, and the FDs apply in every relation that contains all relevant variables. Dummy variables are added to adjust to the change in arity in case of self-joins.

\begin{definition}[FD-Extended Query]\label{definition:fdextended}
	Let $Q(\vec{p}) \leftarrow R_1(\vec{v_1}), \dots, R_m(\vec{v_m})$ be a CQ
	over a schema $\calS=(\calR,\Delta)$. 
	We define two types of extension steps:
	\begin{itemize}
		\item The extension of an atom $R_i(\vec{v_i})$ according to an FD $R_j\colon A\rightarrow b$.\\
		Prerequisites: $\vec{v_j}[A]\subseteq \vec{v_i}$ and $\vec{v_j}[b]\notin \vec{v_i}$.\\
		Effect: The arity of $R_i$ increases by one, $R_i(\vec{v_i})$ is replaced by $R_i(\vec{v_i},\vec{v_j}[b])$.
		In addition, every $R_k(\vec{v_k})$ such that $R_k$=$R_i$ and $k\neq i$ is replaced with $R_k(\vec{v_k},t_k)$, where $t_k$ is a fresh variable in every such step.
		\item The extension of the head $Q(\vec{p})$ according to an FD $R_j\colon A\rightarrow b$.\\
		Prerequisites: $\vec{v_j}[A]\subseteq \vec{p}$ and $\vec{v_j}[b]\notin \vec{p}$.\\
		Effect: The head is replaced by $Q(\vec{p},\vec{v_j}[b])$.
	\end{itemize}
	The {\em FD-extension} of $Q$ is the query $Q^+(\vec{q})\leftarrow R_1^+(\vec{u_m}), \dots, R_m^+(\vec{u_m})$, obtained by performing all possible extension steps on $Q$ according to FDs of $\Delta$ until a fixpoint is reached. The extension is defined over the schema $\calS^+=(\calR^+,\dltplus)$, where $\calR^+$ is $\calR$ with the extended arities, and $\dltplus$ is
\begin{align*}
&\{R_i^+\colon C\to d \mid\exists (R_j\colon A\rightarrow b)\in\Delta, \exists R_i^+(\vec{u_i})\in\atoms(Q^+),
\exists R_j(\vec{v_j})\in\atoms(Q),\\
&\phantom{R_i^+\colon C\to d \mid}\text{ s.t. }
	\vec{u_i}[C]=\vec{v_j}[A]\text{ and }\vec{u_i}[d]=\vec{v_j}[b]\}.
\end{align*}
	\qed
\end{definition}

Given a query, its FD-extension is unique up to a permutation of the added variables and renaming of the new variables.
As the order of the variables and the naming make no difference w.r.t.~enumeration, we can treat the FD-extension as unique.

\begin{example}
	Consider a schema with $\Delta=\{R_1\colon 1\rightarrow 2, R_3\colon 2,3\rightarrow 1\}$, and the query $Q(x)\leftarrow R_1(x,y),R_2(x,z),R_2(u,z),R_3(w,y,z)$.
	As the FDs are $x\rightarrow y$ and $yz\rightarrow w$, the FD-extension is
	\begin{equation*}
	Q^+(x,y)\leftarrow R_1^+(x,y),R_2^+(x,z,y,w),R_2^+(u,z,t_1,t_2), R_3^+(w,y,z).	
	\end{equation*}
	We first apply $x\rightarrow y$ on the head,
	and then $x\rightarrow y$ and consequently  $yz\rightarrow w$ on $R_2(x,z)$.
	These two FDs are now in the schema also for $R_2$, and the FDs of the extension are
	$\dltplus=\{R_1^+\colon 1\rightarrow 2, R_2^+\colon 1\rightarrow 3, R_2^+\colon 3,2\rightarrow 4, R_3^+\colon 2,3\rightarrow 1\}$.
	\qed
\end{example}

We later show that the enumeration complexity of a CQ $Q$ over a schema with FDs only depends on the structure of $Q^+$, which is implicitly given by $Q$ and its schema. Therefore, we introduce the notions of acyclic and free-connex queries for FD-extensions:

\begin{definition}
	Let $Q$ be a CQ over a schema $\calS=(\calR, \Delta)$, and let $Q^+$ be its FD-extension.
	\begin{itemize}
		\item We say that $Q$ is {\em FD-acyclic}, if $Q^+$ is acyclic.
		\item We say that $Q$ is {\em FD-free-connex}, if $Q^+$ is free-connex.
		\item We say that $Q$ is {\em FD-cyclic}, if $Q^+$ is cyclic.
	\end{itemize}
\end{definition}

The following proposition shows that the classes of acyclic queries and free-connex queries are both closed under constructing FD-extensions.

\begin{proposition}\label{prop:plusstructure}
	Let $Q$ be a CQ over a schema $\calS=(\calR,\Delta)$.
	\begin{itemize}
		\item If the query $Q$ is acyclic, then it is FD-acyclic.
		\item If the query $Q$ is free-connex, then it is FD-free-connex.
	\end{itemize}
\end{proposition}

\begin{proof}
	We prove that if $Q$ is acyclic, then $Q^+$ is also acyclic, by constructing a join-tree of $\calH(Q^+)$ given one of $\calH(Q)$. The same proof can be applied to a join tree containing the head to show that if $Q$ is free-connex, then so is $Q^+$.
	Denote by $Q=Q_0,Q_1,\ldots,Q_n=Q^+$ a sequence of queries such that $Q_{i+1}$ is the result of extending all possible relations of $Q_i$ according to a single FD $\delta\in\Delta$. By induction, it suffices to show that if $\calH(Q_i)$ has a join tree, then $\calH(Q_{i+1})$ has one too.
	So consider an acyclic query $Q_i(\vec{p}) \leftarrow R_1(\vec{v}_1), \dots, R_m(\vec{v}_m)$ extended to the query $Q_{i+1}(\vec{q}) \leftarrow R_1(\vec{u}_1), \dots, R_m(\vec{u}_m)$ according to the FD $\delta=R_j\colon \vec{x}\rightarrow y$, and let $T_i=(V_i,E_i)$ be a join tree of $\calH(Q_i)$.
	We claim that the same tree (but with the extended atoms), is a join tree for $Q_{i+1}$. Formally, define $T_{i+1}=(V_{i+1},E_{i+1})$ such that $V_{i+1}=\{R_k(\vec{u}_k)\mid 1\leq k\leq m\}$ and $E_{i+1}=\{(R_k(\vec{u}_k),R_l(\vec{u}_l))\mid(R_k(\vec{v}_k),R_l(\vec{v}_l))\in E_i\}$. Next we show that the running intersection property holds in $T_{i+1}$, and therefore it is a join tree of $\calH(Q_{i+1})$.
	
	Every new variable introduced in the extension appears only in one atom, so the subtree of $T_{i+1}$ containing such a variable contains one node and is trivially connected.
	For any other variable $w\neq y$, the attribute $w$ appears in the same atoms in $Q$ and $Q^+$. Therefore, the subgraph of $T_{i+1}$ containing $w$ is isomorphic to the subgraph of $T_{i}$ containing $w$, and since $T_i$ is a join tree, it is connected.
	It is left to show that the subtree of $T_{i+1}$ containing $y$ is connected. 
	Since $R_j$ is an atom in $Q$ containing $\delta$, it corresponds to vertices in $T_i$ and $T_{i+1}$ containing $\vec{x}\cup\{y\}$.
	Let $R_k$ be some vertex in $T_{i+1}$ containing $y$.
	We will show that all vertices $S_1,\ldots,S_r$ on the path between $R_k$ and $R_j$ contain $y$.
	If $y$ appears in the vertex $R_k$ in $T_i$, then it also appears in $S_1,\ldots,S_r$ since $T_i$ is a join tree. Since the extension doesn't remove occurrences of variables, $y$ appears in these vertices in $T_{i+1}$ as well.
	Otherwise, $y$ was added to $R_k$ via $\delta$, so $R_k$ contains $\vec{x}$. Since $T_i$ is a join tree, the vertices $S_1,\ldots,S_r$ all contain the variables $\vec{x}$. Thus by the definition of $Q_{i+1}$, $y$ is added to each of $S_1,\ldots,S_r$ (if it was not already there) in $T_{i+1}$. Thus also the subtree of $T_{i+1}$ containing $y$ is connected. Therefore $T_{i+1}$
	is indeed a join tree.
\end{proof}

Example~\ref{example:fd-dif1} shows that the converse of the proposition above does not hold.
This means that, by Theorem~\ref{theorem:originalDichotomy}, there are queries $Q$ such that evaluating $Q^+$ is in $\DelayClin$, but evaluating $Q$ cannot be done with the same complexity if we do not assume the FDs.
The following theorem shows that, when relying on the FDs, evaluating $Q^+$ is equally hard to evaluating $Q$.

\begin{theorem}\label{thm:qpEquiv}
	Let $Q$ be a \change{self-join free}\footnote{\change{In previous version of this work, this theorem was stated in general also for CQs with self-joins, but there was a mistake in the proof in the case of self-joins.}} CQ over a schema $\calS=(\calR,\Delta)$, and let $Q^+$ be its FD-extended query. Then, $\qEnum{Q}{\Delta}\equiv_e\qEnum{Q^+}{\dltplus}$.
\end{theorem}

\begin{proof}
	Consider a query $Q(\vec{p}) \leftarrow R_1(\vec{v}_1), \dots, R_m(\vec{v}_m)$ and
	its FD-extension
	$Q^+(\vec{q})\leftarrow R_1^+(\vec{u}_1), \dots, R_m^+(\vec{u}_m)$. We show the two parts of the equivalence.
	
	\begin{claim}
	$\qEnum{Q}{\Delta}\leq_e\qEnum{Q^+}{\dltplus}$.
	\end{claim}
	
	{\it Construction.}
	Given an instance $I$ for $\qEnum{Q}{\Delta}$, we construct an instance $\sigma(I)$ for $\qEnum{Q^+}{\dltplus}$ with two phases: cleaning and extension.
	In the cleaning phase, we remove tuples that interfere with the extended dependencies.
	For every dependency
	$\delta=R_j\colon X\rightarrow y$ and every atom $R_k(\vec{v}_k)$ that contains
	the corresponding variables (i.e., $X\cup\{y\}\subseteq \vec{v}_k$), 
	we correct $R_k$ according to $\delta$:
	we only keep tuples of $R_k^I$ that agree with some tuple of $R_j^I$ over the values of
	$X\cup\{y\}$.
	We denote the cleaned instance by $I_0$.
	The cleaning phase can be done in linear time by first sorting both $R_j^I$ and $R_k^I$ according to
	$X\cup\{y\}$, and then performing one scan over both of them.
	Next, we perform the extension phase. We follow the extension of the schema as described in Definition~\ref{definition:fdextended} and extend the instance accordingly. This phase results in a sequence of instances $I_0,I_1,\ldots,I_n=\sigma(I)$ that correspond to a sequence of queries $Q=Q_0,Q_1,\ldots,Q_n=Q^+$ such that each query is the result of extending an atom or the head of the previous query according to an FD. If in step $i$ the head was extended, we set $I_{i+1}=I_i$. Now assume some relation $R_k$ is extended according to some FD $R_j\colon X\rightarrow y$.
	For each tuple $t\in R^{I_i}_k$, if there is no tuple $s\in R^{I_i}_j$ that agrees with $t$ over the values of $X$, then we remove $t$ altogether.
	Otherwise, we copy $t$ to $R^{I_{i+1}}_k$ and assign $y$ with the same value that $s$ assigns it.
	The extension phase takes linear time for each step. Since the number of FDs is constant in data complexity, the overall construction takes linear time.
	Note that this construction ensures that the extended dependencies hold in $\sigma(I)$.
	Given an answer $\mu|_{\free(Q^+)}\in Q^+(\sigma(I))$, we set $\tau(\mu)=\mu|_{\free(Q)}$. This projection only requires constant time.
	
	\medskip
	\noindent
	{\it Correctness.}
	We now show that $Q(I)=\{\mu|_{\free(Q)}:\mu|_{\free(Q^+)}\in Q^+(I^+)\}$.
	First, if $\mu|_{\free(Q^+)}$ is an answer of $Q^+(I^+)$, then $\mu$ is a homomorphism from $Q^+$ to $I^+$. Since all tuples of $I^+$ appear (perhaps projected) in $I$, then $\mu$ is also a homomorphism from $Q$ to $I$, and $\mu|_{\free(Q)}\in Q(I)$.
It is left to show the opposite direction: if $\mu|_{\free(Q)}\in Q(I)$ then $\mu|_{\free(Q^+)}\in Q^+(I^+)$.
We show by induction on $Q=Q_0,Q_1,\ldots,Q_n=Q^+$ that $\mu|_{\free(Q_i)}\in Q_i(I_i)$.
\change{For the induction base, we claim that the cleaning phase does not remove ``useful'' tuples.
Since $\mu|_{\free(Q)}\in Q(I)$,
there exist tuples, one of each atom of the query, that agree on the values of $X\cup\{y\}$ (these tuples assign the variables with the values $\mu$ assigns them).
Since $Q$ is self-join free, these tuples are not removed during cleaning because every variable-wise dependency agrees with them.\footnote{\change{This does not holds when the CQ contains self-joins, as the following example demonstrates. Consider the query $Q(x,y,z) \leftarrow R(x,y), S(x,y), R(x,z)$ with the dependency $S:1\rightarrow 2$, and the database instance $I$ with $R^I=\{(a,b),(a,c)\}$ and $S^I=\{(a,b)\}$. The cleaning phase removes $(a,c)$ from $R^I$, and so this construction fails.}}
Therefore, $\mu|_{\free(Q)}\in Q(I_0)$.}
Next assume that $\mu|_{\free(Q_i)}\in Q_i(I_i)$, and we want to show that $\mu|_{\free(Q_{i+1})}\in Q_{i+1}(I_{i+1})$. This claim is trivial in case the head was extended.
Note also that there cannot be two distinct answers $\mu|_{\free(Q_{i+1})}$ and $\mu'|_{\free(Q_{i+1})}$ in $Q_{i+1}(I_{i+1})$ such that $\mu|_{\free(Q_{i})}=\mu'|_{\free(Q_{i})}$, as the added variable is bound by the FD to have only one possible value.
Now consider the case where an atom $R_k(\vec{v}_k)$ was extended according to an FD $R_j\colon X\rightarrow y$ since $X\subseteq\vec{v}_k$. The tuple $\mu(\vec{v}_k)\in R_k^{I_i}$ was extended with the value $\mu(y)$ due to the tuple $\mu(\vec{v}_j)\in R_j^{I_i}$ that agrees with it on the values of $X$, and so $\mu(\vec{v}_k,y)\in R_k^{I_{i+1}}$. In case of self-joins, other atoms with the relation $R_k$ are extended with a new and distinct variable. Such variables will be mapped to this value $\mu(y)$ as well. Overall, we have that $\mu$ (extended by mappings of the fresh variables) is also a homomorphism in $Q_{i+1}(I_{i+1})$.

\begin{claim}
	$\qEnum{Q^+}{\dltplus}\leq_e\qEnum{Q}{\Delta}$.
\end{claim}

\noindent
{\it Construction.}
	Given an instance $I^+$ for $\qEnum{Q^+}{\dltplus}$, we construct an instance $\sigma(I^+)$ for $\qEnum{Q}{\Delta}$ with three phases: cleaning, building a lookup table and projection.
In the cleaning phase, we remove tuples that do not contribute to the answer set of $Q^+$ in order to prevent additional answers from appearing in $Q$ after the projection.
This can be seen as unifying the restrictions of different FDs in $\dltplus$ that originate in the same FD in $\Delta$.
For every FD $R_j\colon X\rightarrow y$ in $\Delta$ and every atom $R_k^+(\vec{u}_k)$ such that $X\cup\{y\}\subseteq\vec{u}_k$, we remove all tuples $t\in R_k^{+I^+}$ that agree with some tuple $s\in R_j^{+I^+}$ over $X$ but disagree with $s$ over $y$. 
The cleaning phase can be done in linear time by first sorting both $R_k^{+I^+}$ and $R_j^{+I^+}$ according to $X$.
	Next, we construct a lookup table $T$ to later reconstruct the assignments to $\free(Q^+)\setminus\free(Q)$. For every $y\in\free(Q^+)\setminus\free(Q)$ added to the head due to an FD $R_j\colon X\rightarrow y$, denote by $\vec{x}$ a vector containing the variables of $X$ in lexicographic order.	
	For every tuple in $R_j^{+I^+}$ that assigns $y$ and $\vec{x}$ with the values $y_0$ and $\vec{x}_0$ respectively, we set $T(\vec{x},\vec{x}_0,y)=y_0$. Note that due to the FD, a key cannot map to two different values.
	We conclude the construction by projecting the relations of $I^+$ according to the schema of $Q$.
	These steps result in the construction of an instance $\sigma(I^+)$ and a lookup table $T$ in linear time.
Note that $\Delta$ hold in $\sigma(I^+)$ since $\dltplus$ contains them.

Given $\mu|_{\free(Q)}\in Q(I)$, we define $\tau(\mu|_{\free(Q)})=\mu|_{\free(Q)}\cup\nu_\mu$, where the mapping $\nu_\mu:\free(Q^+)\setminus\free(Q)\rightarrow\idom$ uses the lookup table: For every $y\in\free(Q^+)\setminus\free(Q)$ added due to some FD $R_j\colon X\rightarrow y$, we set $\nu_\mu(y)=T[(\vec{x},\mu(\vec{x}),y)]$.
	Note that $\tau$ is computable in constant time since we use the lookup table $|\free(Q^+)\setminus\free(Q)|$ times, and each access takes constant time.
	
	\medskip
	\noindent
	{\it Correctness.}
	We first claim that the lookup table succeeds in reconstructing the values for the missing head variables: if $\mu|_{\free(Q)}\in Q(\sigma(I^+))$, then $\tau(\mu|_{\free(Q)})=\mu|_{\free(Q^+)}$.
	By definition, for every $y\in\free(Q)$, $\tau(\mu(y))=\mu(y)$. We need to show the same for $y\in\free(Q^+)\setminus\free(Q)$. In this case, $y$ was added to the head due to some FD $R_j\colon X\rightarrow y$, and $\tau(\mu(y))$ is defined to be $\nu_\mu(y)=T[(\vec{x},\mu(\vec{x}),y)]$. 
	Since $\mu$ is a homomorphism into $\sigma(I^+)$, there exists some tuple in $R_j^{\sigma(I^+)}$ that assigns $\vec{x}$ and $y$ with $\mu(\vec{x})$ and respectively $\mu(y)$. This tuple is a projection of a tuple in $R_j^{I^+}$ that assigns $\vec{x}$ and $y$ with the same values.
	Due to this tuple, when constructing the lookup table, we set $T[(\vec{x},\mu(\vec{x}),y)]=\mu(y)$.

	We now show that $Q^+(I^+)=\{\mu|_{\free(Q^+)}:\mu|_{\free(Q)}\in Q(\sigma(I^+))\}$.
	\change{We start by showing that if $\mu|_{\free(Q^+)}\in Q^+(I^+)$, then $\mu|_{\free(Q)}\in Q(\sigma(I^+))$.
	Let $\mu$ be a homomorphism from the variables of $Q^+$ to $I^+$. Then, for every atom $R_i^+(\vec{u}_i)$ of $Q^+$, we have that $R_i^{+I^+}$ contains the tuple $\mu(\vec{u}_i)$, and these tuples all agree on all variables of $Q^+$. In particular, for every FD $X\rightarrow y$ in $\dltplus$, these tuples agree on the $y$ value. Therefore, none of these tuples are removed in the cleaning phase when constructing $\sigma(I^+)$.
	\footnote{\change{This does not hold when the CQ contains self-joins, as the following example demonstrates. Consider the query $Q(v,w,x,y,z) \leftarrow R(x,y,z), R(v,w,x), S(x,y)$ with the dependency $S:1\rightarrow 2$, and the database instance $I^+$ with $R^{I^+}=\{(a,b,c,d),(e,f,a,b)\}$ and $S^I=\{(a,b),(e,g)\}$. The cleaning phase removes $(e,f,a,b)$ from $R^{I+}$, and so this construction fails.}}
    After projecting the cleaned relations of $I^+$ into those of $\sigma(I^+)$, the projections of these tuples appear in $\sigma(I^+)$. Thus, for every  $R_i(\vec{v}_i)$ in $Q$, we have that $R_i^{\sigma(I^+)}$ contains the tuple $\mu(\vec{v}_i)$, and so $\mu$ is a homomorphism from $Q$ to $\sigma(I^+)$. In other words, $\mu|_{\free(Q)}\in Q(\sigma(I^+))$.}
   
	\change{For the second direction we need to show that if $\mu|_{\free(Q)}\in Q(\sigma(I^+))$, then $\mu|_{\free(Q^+)}\in Q^+(I^+)$.
	Let $\mu$ be a homomorphism from $Q$ to $\sigma(I^+)$.
	Thus, for every  $R_i(\vec{v}_i)$ in $Q$, we have that $R_i^{\sigma(I^+)}$ contains the tuple $\mu(\vec{v}_i)$.
	By construction of $\sigma(I^+)$, these tuples are projections of tuples in the cleaned $I^+$. 
	Consider an atom $R_i^+(\vec{u}_i)$ of $Q^+$. Since $R_i^{\sigma(I^+)}$ contains $\mu(\vec{v}_i)$, we have that the cleaned $R_i^{+I^+}$ contains a 
	tuple $s_i$ whose projection into $\vec{v}_i$ is $\mu(\vec{v}_i)$.
	We claim that $s_i=\mu(\vec{u}_i)$. That is, $s_i$ agrees with $\mu$ also on the new attributes.
	Indeed, if $R_i$ was extended due to an FD $R_j\colon X\rightarrow y$, then we know that $\mu(\vec{v}_j)\in R_j^{\sigma(I^+)}$, and that $\mu(\vec{v}_j)$ and $s_i$ must agree on $y$ (in addition to agreeing on $X$), otherwise this $s_i$ would have been deleted in the cleaning phase. Therefore, $s_i$ assigns $y$ with $\mu(y)$.S
	Then, for every atom $R_i^+(\vec{u}_i)$ of $Q^+$, we have that $R_i^{+I^+}$ contains the tuple $\mu(\vec{u}_i)$, and we conclude that $\mu|_{\free(Q^+)}\in Q^+(I^+)$.}
\end{proof}

As described in Corollary~\ref{cor:posExtension}, 
a direct consequence of Theorem~\ref{thm:qpEquiv} is that FD-extensions can be used to expand tractable enumeration classes.
This is due to the fact that $\qEnum{Q}{\Delta}\leq_e\qEnum{Q^+}{\dltplus}$ \change{for self-join free CQs}.
The opposite direction is used in Section~\ref{sec:Positive} to show the lower bounds required for a dichotomy.
\change{To apply a similar conclusion also to CQs with self-joins, we can consider a self-join free variant with the same structure.
Let $\SJF$ be a function that assigns each atom with a different relation symbol. For example, if an atom $R(\vec{v})$ appears in the CQ for the $k$th time, we can replace it by the atom $R_k(\vec{v})$ where $R_k$ is a new relation symbol.
We denote the transformed CQ by $\SJF(Q)$ and we also replace the symbols in the dependency set accordingly to obtain $\SJF(\Delta)$. }

\begin{corollary}\label{cor:posExtension}
	Let $\calC$ be an enumeration class that is closed under exact reduction, $Q$ be a CQ, and $Q^+$ be its FD-extension.
	If \change{$\qEnum{\SJF(Q)^+}{\SJF(\Delta)_{Q^+}}$ is in $\calC$}, then $\qEnum{Q}{\Delta}$ is in $\calC$ too.
\end{corollary}

\begin{proof}
\change{Whenever the self-join free version of a CQ is tractable, the original CQ is tractable too, as we can duplicate the original relations to construct relations for the new distinct symbols and get the same result set. Formally, this proves that $\qEnum{Q}{\Delta}\leq_e\qEnum{\SJF(Q)}{\SJF(\Delta)}$.
	According to Theorem~\ref{thm:qpEquiv}, $\qEnum{\SJF(Q)}{\SJF(\Delta)}\leq_e\qEnum{\SJF(Q)^+}{\SJF(\Delta)_{Q^+}}$.
	By combining these two facts, we get that $\qEnum{Q}{\Delta}\leq_e\qEnum{\SJF(Q)^+}{\SJF(\Delta)_{Q^+}}$. 
	Since $\calC$ is closed under exact reduction, if $\qEnum{\SJF(Q)^+}{\SJF(\Delta)_{Q^+}}\in\calC$, then $\qEnum{Q}{\Delta}\in\calC$.}
\end{proof}

\change{
Since free-connex queries are in $\DelayClin$ and $\DelayClin$ is closed under exact reduction, we immediately get that self-join free CQs that are FD-free-connex are in $\DelayClin$. To conclude the same for CQs with self-joins, we should note that self-joins are irrelevant for the free-connex property, as the following proposition states. 
}

\change{
\begin{proposition}\label{prop:sj-free-connex}
$Q^+$ is free-connex iff $\SJF(Q)^+$ is free-connex.
\end{proposition}
\begin{proof}
First note that $Q$ and $\SJF(Q)$ differ only on the relation names, but they have the same sets of variables in their atoms.
Then, note that the extension procedure mostly depends only on the variable-sets inside the atoms and the matching dependencies, but it has a small sensitivity to the relation names: in the case of self-joins, additional fresh variables are added to the extension; however, every such variable only appears in one atom.
So the difference between $Q^+$ and $\SJF(Q)^+$ is only in the relation names and the fact that atoms in $Q^+$ may have additional variables, where each such variable appears only in one atom.
Now note the following properties of free-connexity: (1) it is not affected by relation names, and (2) it is not affected by the addition or removal of a variable that appears only in a single atom. Both these properties hold because a join tree for a CQ is also a valid join tree for a transformed CQ where the relation names are changed, fresh variables are added in single atoms, or variables that appear only in one atom are removed.
\end{proof}
}

\change{
We can now conclude the same corollary for all CQs.
}

\begin{corollary}\label{cor:realpositiveresult}
	Let $Q$ be a CQ over a schema $\calS=(\calR,\Delta)$.
	If $Q$ is FD-free-connex, then $\qEnum{Q}{\Delta}\in\DelayClin$.
\end{corollary}
\begin{proof}
\change{
	Since $Q^+$ is free-connex, and according to Proposition~\ref{prop:sj-free-connex}, $\SJF(Q)^+$ is free-connex too. 
	Given an instance over a schema with FDs, the same instance is also valid over a similar schema without FDs, so using the identity mapping shows that $\qEnum{\SJF(Q)^+}{\SJF(\Delta)_{Q^+}}\leq_e\qEnum{\SJF(Q)^+}{\emptyset}$. So, by Theorem~\ref{theorem:originalDichotomy}, we have that $\qEnum{\SJF(Q)^+}{\SJF(\Delta)_{Q^+}}\in\DelayClin$.
	 Following Corollary~\ref{cor:posExtension}, $\qEnum{Q}{\Delta}\in\DelayClin$.
	 }
\end{proof}

We can now revisit Example~\ref{example:fd-dif1}.
The query $Q(x,y)\leftarrow R_1(z,x), R_2(z,y)$ is not free-connex. 
Therefore, ignoring the FDs, according to Theorem~\ref{theorem:originalDichotomy} it is not in $\DelayClin$. 
However, given $R_2\colon z\rightarrow y$, the FD-extended query is $Q^+(x,y)\leftarrow R_1^+(z,y,x), R_2^+(z,y)$. As it is free-connex, evaluating $Q^+$ is in $\DelayClin$ by Corollary~\ref{cor:realpositiveresult}.

\section{A Dichotomy for Acyclic CQs} \label{sec:Positive}

In this section, we characterize which self-join free FD-acyclic queries are in $\DelayClin$. We use the notion of FD-extended queries defined in the previous section to establish a dichotomy stating that enumerating the answers to an acyclic query is in $\DelayClin$ iff the query is FD-free-connex. The positive case for the dichotomy is described in Corollary~\ref{cor:realpositiveresult}, and this section concludes the negative case. We prove the following theorem:

\begin{theorem}\label{thm:acyclicDichotomy}
	Let $Q$ be a self-join free FD-acyclic CQ over the schema $\calS=(\calR,\Delta)$.
	\begin{itemize}
		\item If $Q$ is FD-free-connex, then $\qEnum{Q}{\Delta}\in\DelayClin$.
		\item If $Q$ is not FD-free-connex, then $\qEnum{Q}{\Delta}\not\in\DelayClin$, assuming that 
		the product of two $n \times n$ boolean matrices cannot be computed in time $\calO(n^2)$.
	\end{itemize}	
\end{theorem}

Note that the restriction of considering only self-joins-free queries is required only for the negative side. This assumption is standard~\cite{bdg:dichotomy,bb:thesis,Kimelfeld:2012:DCD}, as it allows to assign different atoms with different relations independently. 
The hardness result described here builds on that of Bagan et al.~\cite{bdg:dichotomy} for databases
that are assumed not to have FDs, and it relies on the hardness of \emph{the boolean matrix multiplication problem}.
This problem is defined as the evaluation $\qEnum{\Pi}{\emptyset}$ of the query
$\Pi(x,y)\leftarrow A(x,z),B(z,y)$ over the schema $(\{A,B\},\emptyset)$ where $A,B\subseteq\{1,\ldots,n\}^2$.
It is strongly conjectured that this problem is not computable in $\calO(n^2)$ time, 
and the best currently known algorithms require $\calO(n^\omega)$ time for
some $2.37< \omega< 2.38$~\cite{MatrixMultiplication,MatrixMultiplication2}.

The original proof describes an exact reduction $\qEnum{\Pi}{\emptyset}\leq_e\qEnum{Q}{\emptyset}$.
Since $Q$ is acyclic but not free-connex, it contains a head-path $(x,z_1,\ldots,z_k,y)$.
For a given an instance of the
matrix multiplication problem, an instance of $\qEnum{Q}{\emptyset}$ is constructed, where the variables
$x$,$y$ and $z_1,\ldots,z_k$ of the head-path encode the variables $x$, $y$ and $z$ of $\Pi$,
respectively. All
other variables of $Q$ are assigned constants. This way, $A$ is encoded by an atom containing $x$ and
$z_1$, and $B$ is encoded by an atom containing $z_k$ and $y$. Atoms containing some $z_i$ and
$z_{i+1}$ propagate the value of $z$. Since $x$ and $y$ are in $\free(Q)$, but $z_i$ are not,
the answers to $Q$ correspond to those of $\Pi$.
As no atom of $Q$ contains both $x$ and $y$, the instance can be constructed in linear time.
Constant delay enumeration for $Q$ following a linear time preprocessing would result in the computation of the answers of $\Pi$ in $\calO(n^2)$ time.

FDs restrict the relations that can be assigned to atoms. This means that the reduction cannot be freely performed on databases with FDs, and the proof no longer holds. The following example illustrates where the reduction fails in the presence of FDs.

\begin{example}
	The CQ from Example~\ref{example:fd-dif1} has the form $Q(x,y)\leftarrow R_1(z,x),R_2(z,y)$
	with the single FD $\Delta=\{R_2\colon z\rightarrow y\}$.
	In the previous section, we show that it is in $\DelayClin$, so the reduction should fail.
	Indeed, it would assign $R_2$ with the same relation as $B$ of the matrix multiplication problem,
	but this may have two tuples with the same $z$ value and different $y$ values.
	Therefore, the construction does not yield a valid instance of $\qEnum{Q}{\Delta}$.\qed
\end{example}
We now provide a modification of this construction to show an exact reduction from $\qEnum{\Pi}{\emptyset}$ to $\qEnum{Q^+}{\dltplus}$. Any violations of the FDs are fixed by carefully picking more variables other than those of the head-path to take the roles of $x$,$y$ and $z$ of the matrix multiplication problem. This is done by introducing the sets $V_x$,$V_y$ and $V_z$ which are subsets of $\var(Q)$. We say that a variable $\beta$ {\em plays the role} of $\alpha$, if $\beta\in V_{\alpha}$.
To clarify the reduction, we start by describing a restricted case, where all FDs are unary. The basic idea in the case of general FDs remains the same, but it requires a more involved construction of the sets $V_\alpha$.

\subsection{Unary Functional Dependencies}

For the unary case, we define the sets $V_x, V_y$ and $V_z$ to hold the variables that iteratively imply $x$, $y$ and some $z_i$, respectively. That is, for $\alpha\in\{x,y,z_1,\ldots,z_k\}$ we set $V_\alpha:=\{\alpha\}$ and apply $V_\alpha:=V_\alpha\cup\{\gamma\in\var(Q)\mid\gamma\rightarrow\beta\in\dltplus\wedge\beta\in V_\alpha\}$ until a fixpoint is reached. We then define $V_z:=V_{z_1}\cup\cdots\cup V_{z_k}$.

\paragraph*{The Reduction.\;}
Let $I=(A^I,B^I)$ be an instance of $\qEnum{\Pi}{\emptyset}$.
We define $\sigma(I)$ by describing the relation $R^I$ for every atom $R(\vec{v})\in\atoms(Q^+)$.
If $\var(R)\cap V_y=\emptyset$, then every tuple $(a,c)\in A^I$ is copied to a tuple in $R^I$. Variables in $V_x$ get the value $a$, variables in $V_z$ get the value $c$, and variables that play no role are assigned a constant $\bot$.
That is, we define
$R^{\sigma(I)}$ to be $\{(f(v_1,a,c),\ldots,f(v_k,a,c))\mid (a,c)\in A^I\}$, where:
\begin{equation*}
f(v_i,a,c) = \left\{
\begin{array}{ll}
a &\text{ if }v_i\in V_x\setminus V_z,\\
c &\text{ if }v_i\in V_z\setminus V_x,\\
(a,c) &\text{ if }v_i\in V_x\cap V_z,\\
\bot &\text{otherwise.}
\end{array}
\right.	
\end{equation*}
If $\var(R)\cap V_y\neq\emptyset$, we show that $\var(R)\cap V_x=\emptyset$,
see the sketch for the proof for the well-definedness of the reduction below.
In this case we define the relation similarly with $B^I$. Given a tuple $(c,b)\in B^I$, the variables of $V_y$ get the value $b$, and those of $V_z$ are assigned with $c$.

\begin{example}
	Consider the query $Q^+(x,y,v)\leftarrow R(u,x,z),S(v,y,z)$ with FDs $\dltplus=\{R\colon u\rightarrow x,R\colon u\rightarrow z,S\colon y\rightarrow v\}$.
	Using the head-path $(x,z,y)$, the reduction sets $V_x=\{x,u\}$, $V_y=\{y\}$ and $V_z=\{z,u\}$.
	Given an instance $I$ of the matrix multiplication problem with relations $A^I$ and $B^I$, every tuple $(a,c)\in A^I$ results in a tuple $((a,c),a,c)\in R^{\sigma(I)}$, and every tuple $(c,b)\in B^I$ results in a tuple $(\bot,b,c)\in S^I$.\qed
\end{example}
We now outline the correctness of this reduction.
\begin{description}
	\item [Well-defined reduction:] For an atom $R$, either we have $\var(R)\cap V_y=\emptyset$ or $\var(R)\cap V_x=\emptyset$.
	That is, no atom contains variables from both $V_x$ and $V_y$.
	Due to the definition of $Q^+$, this atom would otherwise also contain both $x$ and $y$.
	However, they cannot appear in
	the same relation according to the definition of a head-path.
	The reduction is therefore well defined, and it can be constructed in linear time via copy and projection.
	\item [Preserving FDs:] The construction ensures that if an FD $\gamma\rightarrow\alpha$ exists, then $\gamma$ has all the roles of $\alpha$. Therefore, either $\alpha$ has no role 
	and corresponds to the constant $\bot$, or every value that appears in $\alpha$ also appears in $\gamma$.
	In any case, all FDs are preserved.
	\item [1-1 mapping of answers:] If a variable of $V_z$ had appeared in the head of $Q^+$, then by the definition of $Q^+$, some $z_i$ would have been in the head as well. This cannot happen according to the definition of a head-path.
	Therefore, $V_z\cap\free(Q^+)=\emptyset$, and the head only encodes the $x$ and $y$ values of the matrix multiplication problem,
	so two different solutions to $\qEnum{Q^+}{\dltplus}$ must differ in either $x$ or $y$,
	and correspond to different solutions of $\qEnum{\Pi}{\emptyset}$.
	For the other direction, the head necessarily contains the variables $x$ and $y$.
	Therefore, two different solutions to $\qEnum{\Pi}{\emptyset}$ correspond to different solutions of $\qEnum{Q^+}{\dltplus}$.
\end{description}

\subsection{General Functional Dependencies}

Next we show how to lift the idea of this reduction to the case of general FDs.
In the case of unary FDs, 
we ensure that the construction does not violate a given FD $\gamma\rightarrow\alpha$, 
by simply encoding the values of $\alpha$ to $\gamma$.
In the general case, when allowing more than one variable on the left-hand side of an FD $\gamma_1,\ldots,\gamma_k\rightarrow\alpha$,
we must be careful when choosing the variables $\gamma_j$ to which we copy the values of $\alpha$. Otherwise, as the following example shows, we will not be able to construct the instance in linear time.

\begin{example}
	Consider the query
	$Q(x,y)\leftarrow R_1(x,z,t_1),R_2(z,y,t_1,t_2)$
	over a schema with the FD $R_2\colon t_1t_2\rightarrow y$.
	Note that $Q=Q^+$ is acyclic but not free-connex, and that $(x,z,y)$ is a head-path in $\calH(Q^+)$.
	To repeat the idea shown in the unary case and ensure that the FDs still hold,
	the variable on the right-hand side of every FD is encoded to the variables on the left-hand side.
	If we encode $y$ to $t_1$, then $R_1$ would contain the encodings of $x$, $y$ and $z$.
	This means that its size will not be linear in that of the matrix multiplication instance,
	and we cannot hope for linear time construction.
	On the other hand, if we choose to encode $y$ only to $t_2$, the reduction works.
	\qed
\end{example}

In the following central lemma, we describe how to carefully pick the variables to which we assign roles in a way that meets the requirements we need for the reduction.
We prove requirements~\ref{lemmaitem:xyzinVxVyVz} and~\ref{lemmaitem:disjointVz} to guarantee a one-to-one mapping between the results of the two problems. Requirement~\ref{lemmaitem:disjointVxVy} enables linear time construction, while requirement~\ref{lemmaitem:leftsideconsistent} is used to show that all FDs are preserved.
The idea is that we consider the join-tree of $Q^+$ and define a partition of its atoms. We then define $V_x$ and $V_y$ to hold variables that appear only in different parts of the tree, ensuring that no atom contains variables of each. The running intersection property of a join-tree is then used to guarantee that the sets are inclusive enough to correct all FD violations.

\begin{lemma}\label{lemma:reductionsets}
	Let $Q$ be a self-join free CQ over a schema $\calS=(\calR,\Delta)$, such that $Q^+$ is acyclic but not free-connex.
	Further let $(x,z_1,\ldots,z_k,y)$ be a head-path of $Q^+$.
	Then, there exist sets of variables $V_x$, $V_y$ and $V_z$ such that:
	\begin{enumerate}
		\item\label{lemmaitem:xyzinVxVyVz} $x\in V_x$, $y\in V_y$, $\{z_1,\ldots z_k\}\subseteq V_z$.
		\item\label{lemmaitem:disjointVz} $V_z\cap\free(Q^+)=\emptyset$.
		\item\label{lemmaitem:disjointVxVy} For every $R\in\atoms(Q^+)$: $\var(R)\cap V_y=\emptyset$ or $\var(R)\cap V_x=\emptyset$.
		\item\label{lemmaitem:leftsideconsistent} For every $U\rightarrow v\in\dltplus$ s.t. $v\in V_\alpha$ with $\alpha\in\{x,y,z\}$: $U\cap V_\alpha\neq\emptyset$.
	\end{enumerate}
\end{lemma}

\begin{figure}[t]
	\begin{center}
		\begin{tikzpicture}[node distance = 0.07]
		
		
		\node[rectangle, rounded corners, draw, minimum width=38.8em, minimum height=4em]
		(H) at (0,0) {};
		
		
		\node (dotleft1) [left=3.35 of H.north] {};
		\node (dotleft2) [left=3.35 of H.south] {};
		\node (dotright1) [right=3.38 of H.north] {};
		\node (dotright2) [right=3.38 of H.south] {};
		
		
		\node (Tx) [below left = 0 and 4.2 of H.south] {$T_x$};
		\node (Ty) [below right = 0 and 4.4 of H.south] {$T_y$};
		\node (Tz) [below = 0 of H.south] {$T_\text{mid}$};
		
		
		\node[rectangle, rounded corners, draw, thick, inner sep = 3] (e1) [below right= -0.5 and 0.52 of H.west] {$e(x,z_1)$};
		
		\node (d1r) [right = of e1] {$\ldots$};
		\node (d1l1) [left=0.3 of e1] {};
		\node (d1l2) [above left= -0.2 and 0.3 of e1] {};
		\node (d1l3)  [below left= -0.2 and 0.3 of e1] {};
		
		\node[rectangle, rounded corners, draw, thick, inner sep = 3] (vx) [right = of d1r] {$sep_x$\phantom{)}\!\!};
		\node (d2r1) [right= of vx] {$\ldots$};
		\node (d2d1) [below= 0.3 of vx] {};
		\node (d2d2) [below left= 0.3 and -0.2 of vx]  {};
		\node (d2d3) [below right= 0.3 and -0.2 of vx]  {};
		
		\node[rectangle, rounded corners, draw, thick, inner sep = 3] (e2) [right = of d2r1] {$e(z_1,z_2)$};
		\node (d3r1) [right= of e2] {$\ldots$};
		\node (d3d1) [below= 0.3 of e2] {};
		\node (d3d2) [below left= 0.3 and -0.6 of e2]  {};
		\node (d3d3) [below right= 0.3 and -0.6 of e2]  {};
		
		\node[rectangle, rounded corners, draw, thick, inner sep = 3] (e3) [right = of d3r1] {$e(z_{k-1},z_k)$};
		\node (d4r1) [right= of e3] {$\ldots$};
		\node (d4d1) [below= 0.3 of e3] {};
		\node (d4d2) [below left= 0.3 and -0.8 of e3]  {};
		\node (d4d3) [below right= 0.3 and -0.8 of e3]  {};
		
		\node[rectangle, rounded corners, draw, thick, inner sep = 3] (vy) [right = of d4r1] {$sep_y$\phantom{)}\!\!};
		\node (d5r1) [right= of vy] {$\ldots$};
		\node (d5d1) [below= 0.3 of vy] {};
		\node (d5d2) [below left= 0.3 and -0.2 of vy]  {};
		\node (d5d3) [below right= 0.3 and -0.2 of vy]  {};
		
		\node[rectangle, rounded corners, draw, thick, inner sep = 3] (e4) [right = of d5r1] {$e(z_k,y)$};
		\node (d6r1) [right=0.3 of e4] {};
		\node (d6r2) [above right= -0.2 and 0.3 of e4] {};
		\node (d6r3)  [below right= -0.2 and 0.3 of e4] {};

		\draw[dashed] (dotleft1.center) -- (dotleft2.center);
		\draw[dashed] (dotright1.center) -- (dotright2.center);
		
		\draw[thick] (e1) -- (d1r);
		\draw[thick] (e1.west) -- (d1l1);
		\draw[thick] (e1.west) -- (d1l2);
		\draw[thick] (e1.west) -- (d1l3);
		
		\draw[thick] (vx) -- (d1r);
		\draw[thick] (vx) -- (d2r1);
		\draw[thick] (vx.south) -- (d2d1);
		\draw[thick] (vx.south) -- (d2d2);
		\draw[thick] (vx.south) -- (d2d3);
		
		\draw[thick] (e2) -- (d2r1);
		\draw[thick] (e2) -- (d3r1);
		\draw[thick] (e2.south) -- (d3d1);
		\draw[thick] (e2.south) -- (d3d2);
		\draw[thick] (e2.south) -- (d3d3);
		
		\draw[thick] (e3) -- (d3r1);
		\draw[thick] (e3) -- (d4r1);
		\draw[thick] (e3.south) -- (d4d1);
		\draw[thick] (e3.south) -- (d4d2);
		\draw[thick] (e3.south) -- (d4d3);
		
		\draw[thick] (vy) -- (d4r1);
		\draw[thick] (vy) -- (d5r1);
		\draw[thick] (vy.south) -- (d5d1);
		\draw[thick] (vy.south) -- (d5d2);
		\draw[thick] (vy.south) -- (d5d3);
		
		\draw[thick] (e4) -- (d5r1);
		\draw[thick] (e4.east) -- (d6r1);
		\draw[thick] (e4.east) -- (d6r2);
		\draw[thick] (e4.east) -- (d6r3);
		
		\end{tikzpicture}
	\end{center}
	\caption{Join tree $T$ of $\calH(Q^+)$ for head-paths of length greater than $2$. The subtrees $T_x$, $T_y$ and $T_\text{mid}$ are disjoint, and are separated by the nodes $sep_x$ and $sep_y$.}
	\label{fig:generalnegative1}
\end{figure}

\begin{proof}
	We first define a partition of the atoms of $Q$ into two or three sets: $T_x$, $T_y$ and possibly $T_\text{mid}$.
	Let $T$ be a join tree of $\calH(Q^+)$, and denote the hyperedges on the head-path by $e(x,z_1),\ldots,e(z_k,y)$.
	Note that, by definition, each hyperedge of the head-path is a vertex of $T$  and an atom of $Q^+$. By the running intersection property of $T$ and since the path is cordless, we can conclude that there is a simple path $P$ from $e(x,z_1)$ to $e(z_k,y)$ in $T$, such that $e(z_1,z_2),\ldots,e(z_{k-1},z_k)$ lie on that path in the order induced by the head-path.
	Let $sep_x$ be the first atom on the path $P$ that does not contain $x$. This exists because $e(z_k,y)$ does not contain $x$, as the head-path is chordless. Similarly, let $sep_y$ be the last atom on $P$ that does not contain $y$.
	Let $T_x$ be the set of atoms $v$ such that the unique path from $v$ to $e(x,z_1)$ in $T$ does not go through $sep_x$.
	Similarly, let
	$T_y$ be the set of atoms $w$ such that the unique path from $w$ to $e(z_k,y)$ in $T$ does not go through $sep_y$.
	Next set $T_\text{mid}=V(T)\setminus(T_x\cup T_y)$.
	Note that $e(x,z_1)\in T_x$ and $e(z_k,y)\in T_y$, but $T_\text{mid}$ may
	be empty (this happens in the case that the head-path is of length two). By definition, the atoms of $Q^+$ are exactly $T_x\cup T_\text{mid}\cup T_y$, and next we show that this union is disjoint. Figure~\ref{fig:generalnegative1} depicts the established partition.
	
	\begin{claim}
		The sets $T_x$ and $T_y$ are disjoint.
	\end{claim}
	\begin{proof}[Proof of the Claim]
		Assume by contradiction that there is some $v\in T_x\cap T_y$. Let $P_x$ be the unique simple path from $v$ to $e(x,z_1)$, and recall that since $v\in T_x$ it does not go through $sep_x$. Similarly let $P_y$ be the unique simple path from $v$ to $e(z_k,y)$ that does not go through $sep_y$.
		
		We first claim that there exists some atoms $w$ that appears in all three paths $P$, $P_x$ and $P_y$. Take $w$ to be the first atom on $P_x$ that is also in $P$ and set $P_x^w$ to be the simple path from $v$ to $w$.
		Such an atom $w$ exists because the last atom of $P_x$ is $e(x,z_1)$ which is in $P$.
		Further set $P^w$ to be the simple path from $w$ to $e(z_k,y)$.
		Concatenating the paths $P_x^w$ and $P^w$,
		we obtain a simple path from $v$ to $e(z_k,y)$. Since the simple paths in a tree are unique, this is exactly $P_y$, and so $w$ is also in $P_y$.
		
		Our second claim is that if an atom $u$ is in both $P$ and $P_x$, then it contains the variable $x$. Assume by contradiction that such an atom $u$ does not contain $x$. Then $u$ is an atom on $P$ not containing $x$, and by definition of $sep_x$, the simple path from $u$ to $e(x,z_1)$ contains $sep_x$. As this path is a subpath of $P_x$, $P_x$ contains $sep_x$, in contradiction to the fact that $v\in T_x$.
		Similarly, if an atom is in both $P$ and $P_y$, then it contains $y$.
		
		Combining the two claims, we have an atom $w$ containing both $x$ and $y$, in contradiction to the fact that a head-path is chordless by definition. Therefore we conclude that $T_x$ and $T_y$ are indeed disjoint. 
	\end{proof}
	
	Now we are ready to define the sets of variables $V_x, V_y$ and $V_z$.
	We define $V_x$ recursively to contain $x$ and variables that imply those of $V_x$, but without variables that appear outside of $T_x$.
	$V_y$ is defined symmetrically.
	$V_z$ contains $z_1,\ldots,z_k$ and variables that imply those of $V_z$ but without free variables.
	Formally,
	\begin{equation*}
	\mathsf{Implies}(V)=\{u\in\var(Q)\mid \exists U\rightarrow w\in\dltplus 		\text{ with }w\in V\text{ and }u\in U\}
	\end{equation*}
	 for $V\subseteq\var(Q)$, and we define via fixpoint iteration the following:
	\begin{itemize}
		\item $\bm{V_x}$: base $V_x:=\{x\}$; rule $V_x:=(V_x\cup\mathsf{Implies}(V_x))\setminus\var(T_y\cup T_\text{mid})$
		\item $\bm{V_y}$: base $V_y:=\{y\}$; rule $V_y:=(V_y\cup\mathsf{Implies}(V_y))\setminus\var(T_x\cup T_\text{mid})$
		\item$\bm{V_z}$: base $V_z:=\{z_1,\ldots z_k\}$; rule $V_z:=(V_z\cup\mathsf{Implies}(V_z))\setminus\free(Q^+)$
	\end{itemize}

	We now prove that $V_x$, $V_y$ and $V_z$ meet the requirements of the lemma. Requirements~\ref{lemmaitem:xyzinVxVyVz} and~\ref{lemmaitem:disjointVz} follow immediately from the definition of the sets. To prove requirement~\ref{lemmaitem:disjointVxVy}, let $R\in\atoms(Q^+)$. If $R\in T_x$, then by definition of $V_y$ we have that $\var(R)\cap V_y=\emptyset$. Otherwise, $R\in T_y\cup T_\text{mid}$, and similarly $\var(R)\cap V_x=\emptyset$. It is left to show requirement~\ref{lemmaitem:leftsideconsistent}.
	
	Let $\delta = U\rightarrow v\in\dltplus$ where $v\in V_\alpha$. We first show the case of $\alpha=z$. If $U\cap V_z=\emptyset$, then $U\subseteq\free(Q^+)$, and by the definition of $Q^+$, $v\in\free(Q^+)$, which is a contradiction to the definition of $V_z$.
	Now we prove the case where $\alpha= x$. The case $\alpha=y$ is symmetric.
	Denote by $e(U,v)$ an atom containing all variables of $\delta$. As $v\in V_x$, we know that $e(U,v)\notin T_y\cup T_\text{mid}$, therefore $e(U,v)\in T_x$.
	Assume by contradiction that $U\cap V_x = \emptyset$.
	Let $u\in U$. By definition of $V_x$, this means that $u\in\var(e_u)$ for some $e_u\in T_y\cup T_\text{mid}$. 
	As $T_x$, $T_y$ and $T_\text{mid}$ are disjoint, we have that $e_u\notin T_x$, which means that the path between $e_u$ and $e(x,z_1)$ goes through $sep_x$. This means that the path from $e_u$ to $e(U,v)$ goes through $sep_x$ too, otherwise the concatenation of this path with the path from $e(U,v)$ to $e(x,z_1)$ would result in a path from $e_u$ to $e(x,z_1)$ not going through $sep_x$. By the running intersection property, $u\in\var(sep_x)$.
	Since this is true for all $u\in U$,
	it follows that $v\in\var(sep_x)$ by definition of $Q^+$, contradicting the fact that $v\in V_x$.
\end{proof}

With the sets $V_x,V_y,V_z$ at hand, we can now perform the reduction in the presence of general FDs.

\begin{lemma}\label{lemma:acyclicgeneralfdnegative}
	Let $Q$ be a self-join free CQ over a schema $\calS=(\calR,\Delta)$.
	If $Q^+$ is acyclic and not free-connex, then $\qEnum{\Pi}{\emptyset}\leq_e\qEnum{Q^+}{\dltplus}$.
\end{lemma}

\begin{proof}
	Let $I_{A,B}=(A^I,B^I)$ be an instance of $\qEnum{\Pi}{\emptyset}$ over the domain $\calD=\{1,\ldots,n\}$.
	We define an instance $\sigma(I_{A,B})$ of $\qEnum{Q^+}{\dltplus}$ based on the sets $V_x$, $V_y$ and $V_z$ from Lemma~\ref{lemma:reductionsets} and the relations $A^I$ and $B^I$.
	Since $Q^+$ is acyclic but not free-connex, it contains some head-path $(x,z_1,\ldots,z_k,y)$.
	
	To define the instance $\sigma(I_{A,B})$, we first fix the functions $f_A$ and $f_B$:
	\begin{align*}
	f_A(v,a,c) &= \left\{
	\begin{array}{lr}
	a & : v\in V_x\setminus V_z\\
	c & : v\in V_z\setminus V_x\\
	(a,c) & : v\in V_{x}\cap V_z\\
	\bot & : \text{otherwise}
	\end{array}
	\right.,\qquad   
	f_B(v,b,c) &= \left\{
	\begin{array}{lr}
	b & : v\in V_y\setminus V_z\\
	c & : v\in V_z\setminus V_y\\
	(b,c) & : v\in V_{y}\cap V_z\\
	\bot & : \text{otherwise}
	\end{array}  
	\right.
	\end{align*} 
	We partition all relational atoms of $Q^+$ into two 
	sets: $\calR_A^+$ and $\calR_B^+$.
	The set $\calR_A^+$ is defined as $\{R^+\in\atoms(Q^+)\mid \var(R^+)\cap V_y=\emptyset\}$ and $\calR_B^+$ is $\atoms(Q^+)\setminus\calR_A^+$.
	To obtain an instance $\sigma(I_{A,B})$ of $\qEnum{Q^+}{\Delta^+}$, we apply $f_A$ to the atoms in
	$\calR_A^+$ using the values of $A^I$, while the atoms in $\calR_B^+$ use $f_B$ and $B^I$. That is, if $R^+(u_1,\ldots,u_m)\in\calR_A^+$, then $(R^+)^{\sigma(I_{A,B})}$ is defined to be $\{(f_A(u_1,a,c),\ldots,f_A(u_m,a,c))\mid (a,c)\in A^I\}$. Otherwise, $R^+(u_1,\ldots,u_m)$ is in $\calR_B^+$, and $(R^+)^{\sigma(I_{A,B})}=\{(f_B(v_1,b,c),\ldots,f_B(v_p,b,c))\mid (c,b)\in B^I\}$.
	The mapping $\tau$ is defined as the projection onto the variables $x$ and $y$.
	Note that the instance can be constructed in linear time, and the projection can be computed in constant time. 
	
	We now claim that $\sigma(I_{A,B})$ is a database over the schema $(\calR^+,\dltplus)$, as
	all the FDs of $\dltplus$ are satisfied. Let $\delta = R_j^+\colon U\rightarrow v\in\dltplus$.
	If $v\notin V_x\cup V_y\cup V_z$,
	then $\delta$ holds as $v$ is assigned the value $\bot$ in every tuple in $(R_j^+)^{\sigma(I_{A,B})}$.
	Next assume that $v\in V_{x}\setminus V_{z}$.
	By point~\ref{lemmaitem:leftsideconsistent} of Lemma~\ref{lemma:reductionsets},
	there is some $u\in U$ such that $u\in V_x$. Thus in every tuple in $(R_j^+)^{\sigma(I_{A,B})}$, 
	if $v$ is assigned the value $a$, then $u$ is either assigned the value $a$ or $(a,c)$ for some $c\in\{1,\ldots,n\}$
	and in either case $\delta$ is satisfied. The proof for the cases where $v\in V_{z}\setminus (V_{x}\cup V_{y})$ and $v\in V_{y}\setminus V_{z}$ is similar.
	Next assume that $v\in V_{x}\cap V_{z}$.
	By point~\ref{lemmaitem:leftsideconsistent} of Lemma~\ref{lemma:reductionsets},
	there are some $u_1,u_2\in U$ such that $u_1\in V_{x}$ and $u_2\in V_{z}$ and for every tuple in $(R_j^+)^{\sigma(I_{A,B})}$,  if $v$ is assigned the value $(a,c)$, then $u_1$ is either assigned the value $a$ or $(a,c)$, $u_2$ is either assigned the value $c$ or $(a,c)$, and so $\delta$ is satisfied.
	The case $v\in V_{y}\cap V_{z}$ is similar. Note that the case where $v\in V_x\cap V_y$ cannot occur due to point~\ref{lemmaitem:disjointVxVy}
	of Lemma~\ref{lemma:reductionsets}.
	
	The structure of the head-path $(x,z_1\ldots,z_k,y)$
	guarantees that all answers to $\qEnum{Q^+}{\dltplus}$ correspond to those of $\qEnum{\Pi}{\emptyset}$ and vice versa.
	Indeed, let $\mu|_{\{x,y\}}	\in\Pi(I_{A,B})$.
\change{We have that $(\mu(x),\mu(z))\in A^I$ and $(\mu(z),\mu(y))\in B^I$.}	
	We define $\nu_\mu:\var(Q^+)\rightarrow\calD$ as follows.
	\begin{equation*}
	\nu_\mu(v)= \left\{
	\begin{array}{ll}
	\mu(x) & : v\in V_x\setminus V_z,\\
	\mu(y) & : v\in V_y\setminus V_z,\\
	\mu(z) & : v\in V_z\setminus (V_x\cup V_y),\\
	(\mu(x),\mu(z)) & : v\in V_z\cap V_x,\\
	(\mu(y),\mu(z)) & : v\in V_z\cap V_y,\\
	\bot & : \text{otherwise}.
	\end{array}  
	\right.
	\end{equation*}
	By definition, $f_A(v,\mu(x),\mu(z))=\nu_\mu(v)$ and $f_B(v,\mu(y),\mu(z))=\nu_\mu(v)$.
	As $\mu|_{\{x,y\}}\in\Pi(I_{A,B})$, we have that $(\mu(x),\mu(z))\in A^I$ and $(\mu(z),\mu(y))\in B^I$.
	Consider any atom $R^+(u_1,\ldots,u_m)$ of $Q^+$. If $R^+(u_1,\ldots,u_m)\in \calR_A^+$, then $\nu_\mu(u_1,\ldots,u_m)=(f_A(u_1,\mu(x),\mu(z)),\ldots,f_A(u_m,\mu(x),\mu(z)))\in (R^+)^{\sigma(I_{A,B})}$.
	If $R^+(u_1,\ldots,u_m)\in \calR_B^+$, we have that $\var(R^+)\cap V_x=\emptyset$ by point~\ref{lemmaitem:disjointVxVy} of Lemma~\ref{lemma:reductionsets}. Then, 	$\nu_\mu(u_1,\ldots,u_m)=(f_B(u_1,\mu(z),\mu(y)),\ldots,f_B(u_m,\mu(z),\mu(y)))$ is in $(R^+)^{\sigma(I_{A,B})}$.
	Therefore, $\nu_\mu|_{\free(Q^+)}\in Q^+(\sigma(I_{A,B}))$.
	Since $x\in V_x\setminus V_z$ and $y\in V_y\setminus V_z$, we have that $\nu_\mu(x)=\mu(x)$ and $\nu_\mu(y)=\mu(y)$, and so $\tau(\nu_\mu)=\mu|_{\{x,y\}}$.
	Moreover, any answer $\mu'|_{\free(Q^+)}\in Q^+(\sigma(I_{A,B}))$ that has $\tau(\mu'|_{\free(Q^+)})=\mu|_{\{x,y\}}$ assigns $\mu(x)$ to variables in $V_x\setminus V_z$, assigns $\mu(y)$ to variables in $V_y \setminus V_z$ and assigns $\bot$ to variables not in $V_x\cup V_y\cup V_z$. By points \ref{lemmaitem:disjointVz} and \ref{lemmaitem:disjointVxVy} of Lemma~\ref{lemma:reductionsets}, $V_z\cap \free(Q^+)=\emptyset$ and $V_x\cap V_y=\emptyset$, so these are the only variables in $\free(Q^+)$. Therefore $\mu'|_{\free(Q^+)}=\nu_\mu|_{\free(Q^+)}$.
	
	Next assume that $\mu'|_{\free(Q^+)}\in Q^+(\sigma(I_{A,B}))$. Let $R_x^+$ be an atom containing $x$ and $z_1$ (such an atom exists by the definition
	of the head-path).
	By point~\ref{lemmaitem:disjointVxVy} of Lemma~\ref{lemma:reductionsets} we have $\var(R^+)\cap V_y=\emptyset$, and $R^+\in \calR_A^+$.
	By points \ref{lemmaitem:xyzinVxVyVz} and \ref{lemmaitem:disjointVz} of Lemma~\ref{lemma:reductionsets} and since $x$ is a free variable, we have $x\in V_x\setminus V_z$ and $z_1\in V_z$.
	Thus there exists some $(a,c)\in A^I$ such that $\mu'(x)=f_A(x,a,c)=a$ and $\mu'(z_1)=f_A(z_1,a,c)\in\{a,(a,c)\}$.
	Similarly, there exists some $(c',b)\in B^I$ such that $\mu'(y)=f_B(y,b,c')=b$ and $\mu'(z_k)=f_B(z_k,b,c')\in\{c',(b,c')\}$.
	It remains to show that $c=c'$.
	We show by induction on $i$ that $\mu'(z_i)$ is either $c$ or of the form $(t,c)$ with some value $t$.
	We know this fact for $i=1$ since this is how we define $c$.
	If $k>1$, then consider an atom $R_i^+(\vec{v}_i)$ containing $\{z_{i-1},z_i\}$. Then $\mu'$ maps $\vec{v}_i$ to some tuple $t\in R_i^+$ that assigns $z_{i-1}$ with a value of the form $c$ or $(t,c)$. Since $z_i \in V_z$, $t$ also assigns $z_i$ with a value of such a form, so $\mu'(z_i)$ is of the form $c$ or $(t,c)$ too. This show that $c=c'$.
	Therefore, $\tau(\mu')\in\Pi(I_{A,B})$. Moreover, since $\tau$ is simply a projection, $\tau(\mu')$ is uniquely defined.
\end{proof}

By combining Theorem~\ref{thm:qpEquiv} and Lemma~\ref{lemma:acyclicgeneralfdnegative}, we have the exact reduction $\qEnum{\Pi}{\emptyset}\leq_e\qEnum{Q^+}{\dltplus}\leq_e\qEnum{Q}{\Delta}$. Therefore having $\qEnum{Q}{\Delta}$ in $\DelayClin$ would mean that $\qEnum{\Pi}{\emptyset}\in\DelayClin$, which contradicts the conjecture of the lower bound for matrix multiplication. This concludes the proof of Theorem~\ref{thm:acyclicDichotomy}.
Note that Theorem~\ref{thm:acyclicDichotomy} does not contradict the dichotomy of Theorem~\ref{theorem:originalDichotomy}: if for a given query $Q$ we have that $Q^+$ is acyclic but not free-connex, then $Q$ is not free-connex by Proposition~\ref{prop:plusstructure}.

\section{Cyclic CQs} \label{sec:cyclic}

In the previous section, we established a classification of FD-acyclic CQs, but we did not
consider FD-cyclic queries. A known result states that, under certain assumptions, self-join free cyclic
queries are not in $\DelayClin$~\cite{bb:thesis}. In this section, we therefore explore how
FD-extensions can be used to obtain some insight on the implications of this result in the
presence of FDs. We show that (under the same assumptions) self-join free FD-cyclic queries that contain
only unary FDs cannot be evaluated in linear time. For schemas containing only unary FDs,
this extends the dichotomy presented in the previous section to all CQs,
and also proves a dichotomy for the queries that can be enumerated in linear delay.
We will prove the following theorem using the $\tetra{k}$ problem, which is defined below.

\begin{theorem}\label{thm:cyclicHardness}
	Let $Q$ be a self-join free CQ over a schema $\calS=(\calR,\Delta)$, where $\Delta$ only contains unary FDs. If $Q$ is FD-cyclic, then $\qDecide{Q}{\Delta}$ cannot be solved in linear time,
	assuming that the $\tetra{k}$ problem cannot be solved in linear time for any $k\geq3$.
\end{theorem}

As before, the initial hardness proof for cyclic queries no longer holds in the presence of FDs, and we present a modified reduction that satisfies the FDs.
We start by describing the hypothesis used to obtain the conditional lower bounds.
$\tetra{k}$ is defined as the hypergraph with the vertices $\{1,\ldots,k\}$ and the edges $\{\{1,\ldots,k\}\setminus\{i\}\mid i\in\{1,\ldots,k\}\}$. 
Let $\calH$ be a hypergraph. With a slight abuse of notation, we also denote by $\tetra{k}$ the 
decision problem of whether $\calH$ contains a subhypergraph isomorphic to $\tetra{k}$.

$\tetra{3}$ is the problem of deciding whether a graph contains a triangle, 
and it is strongly believed not to be solvable within time linear in the size of the graph \cite{DBLP:conf/focs/WilliamsW10}.
The generalization of this assumption is that $\tetra{k}$ cannot be solved in time linear in the size of the graph for any $k\geq 3$. This assumption is stronger than the one used in Section~\ref{sec:Positive}, as $\tetra{3}$ can be reduced to the matrix multiplication problem~\cite{DBLP:conf/focs/WilliamsW10}.
However, we do have reasons to also believe the $\tetra{k}$ assumption for any $k\geq 4$.
The $(\ell,k)$-Hyperclique Hypothesis~\cite{DBLP:conf/soda/LincolnWW18} states that, in a $k$-uniform hypergraph of $n$ vertices, $n^{k-o(1)}$ time is required to find a set of $\ell$ vertices such that each of its subsets of size $k$ forms a hyperedge.
Solving $\tetra{k}$ in linear time would contradict the $(\ell,k)$-Hyperclique Hypothesis. Indeed, an algorithm that decides whether a hypergraph contains a $\tetra{k}$ in linear time runs in at most $n^{k-1}$ many steps, and thus
can also detect a $(k-1)$-hyperclique in a $k$-uniform graph in time $n^{k-1}$, which is less than $n^{k-o(1)}$.

We will show that if $Q^+$ is cyclic and only unary FDs are present,
the problem $\tetra{k}$ for some $k$ can be reduced to $\qDecide{Q^+}{\dltplus}$. In the following definition, $\calH^b$ is a pseudo-minor isomorphic to some $\tetra{k}$. To perform the said reduction, we will use this pseudo-minor on a graph describing our query.

\begin{definition}\label{def:tetpm}
	Let $\calH$ be a cyclic hypergraph. We denote by $\tetpm{\calH}$ the pairs of pseudo-minors $(\calH^a,\calH^b)$ of $\calH$ such that:
	\begin{enumerate}
	\item $\calH^a$ is obtained by a (possibly empty) set of vertex removal and edge removal operations on $\calH$.
	\item $\calH^b$ is obtained by a (possibly empty) set of edge contraction and edge removal operations on $\calH^a$.
	\item $\calH^b$ is isomorphic to $\tetra{k}$ for some $k\geq 3$.
	\item Either $\calH^a=\calH^b$ or $\calH^a$ is a chordless cycle.
	\end{enumerate}
	Given a query $Q$, we define $\tetpm{Q}=\tetpm{\calH(Q)}$.
\end{definition}

Brault-Baron~\cite[Theorem 11]{bb:thesis} showed that a cyclic hypergraph $\calH$ admits some $\tetra{k}$ as a pseudo-minor. We describe the proof here in our terminology.

\begin{lemma}[\cite{bb:thesis}, Theorem 11]\label{lemma:acyclicEquivalence}
	Let $\calH$ be a hypergraph. If $\calH$ is cyclic, then $\tetpm{\calH}$ is non-empty.
\end{lemma}
\begin{proof}
	If $\calH$ has a chordless cycle $C$ as an induced subgraph, then removing vertices not in $C$ followed by performing all possible edge removals results in a chordless cycle $\calH^a$. Then, $\calH^b$ isomorphic to $\tetra{3}$ is obtained by a repeated use of edge-contraction followed by performing all possible edge removals. In this case, $(\calH^a,\calH^b)\in\tetpm{\calH}$.
	If $\calH$ does not contain a chordless cycle, since it is not acyclic, it is non-conformal.
	Consider its smallest non-conformal clique. The clique is not contained in any edge (since it is non-conformal), and it is a $\tetra{k}$ because of its minimality. Therefore, removing all vertices other than the clique, and then performing all possible edge removals, results in a graph $\calH^a=\calH^b$ isomorphic to some $\tetra{k}$. Again, $(\calH^a,\calH^b)\in\tetpm{\calH}$.
\end{proof}

For the reduction we present next, we first need to show that for an FD-cyclic query $Q$, no pseudo-minor in $\tetpm{Q^+}$ contains all variables of any FD $X\rightarrow y$.
In the following we assume that $\Delta$ only contains non-trivial FDs, meaning $y\notin X$.

\begin{lemma}\label{lemma:fd-pm}
	Let $Q$ be a self-join free FD-cyclic CQ over a schema $\calS=(\calR,\Delta)$. Let $(\calH^a,\calH^b)\in\tetpm{Q^+}$ and 
	$\calH^a=(V,E)$. For every non-trivial $X\rightarrow y\in\dltplus$, we have $X\cup\{y\}\not\subseteq V$.
\end{lemma}

\begin{proof}
	We start with an observation regarding the FDs.
	Note that in $\calH(Q^+)$ some edge contains the vertices $X\cup\{y\}$, and by the construction of $Q^+$, every edge that contains $X$ must also contain $y$. 
	These properties still hold after any sequence of vertex removals and edge removals as long as none of the vertices $X\cup\{y\}$ are removed.
	Therefore if none of the vertices $X\cup\{y\}$ were removed, there must be an edge in $\calH^a$ containing $X\cup\{y\}$, and every edge in $\calH^a$ containing $X$ also contains $y$.
	
	We distinguish two cases. In the first case, $\calH^b=\calH^a$ is a $\tetra{k}$ obtained from $\calH(Q^+)$ by a sequence of vertex and edge removals. If $X\cup\{y\}\subseteq V$, then by the definition of $\tetra{k}$ it should contain the edge $V\setminus \{y\}$. Such an edge cannot exist since it contains all of $X$ but not $y$. Therefore, such a $\tetra{k}$ cannot contain all of $X\cup\{y\}$, and in this case we conclude that $X\cup\{y\}\not\subseteq V$.
	In the second case, $\calH^a$ is a cycle, and $\calH^b$ is a $\tetra{3}$ obtained by performing edge contraction steps on it. If none of $X\cup\{y\}$ were removed, some edge of $\calH^a$ must contain all of them. Since all edges are of size $2$ it must be that $|X|=1$. Denote $X=\{x\}$. Since we consider a cycle containing both $x$ and $y$, there should be at least one edge containing $x$ but not containing $y$. Since we showed it is not possible, such a cycle cannot contain all of $X\cup\{y\}$.
\end{proof}

We are now ready to establish the reduction. 
Given a pseudo-minor of $\tetpm{Q^+}$ isomorphic to some $\tetra{k}$, we can reduce the problem of checking whether a hypergraph contains a subhypergraph 
isomorphic to $\tetra{k}$ to finding a boolean answer to $Q^+$.

\begin{lemma}\label{lemma:simple-pm}
	Let $Q$ be a self-join free FD-cyclic CQ over a schema $\calS=(\calR,\Delta)$, where $\Delta$ only contains unary FDs. 
	Let $(\calH^a,\calH^b)\in\tetpm{Q^+}$ such that $\calH^b$ is isomorphic to $\tetra{k}$ for some $k$. Then, there is a linear time reduction $\tetra{k}\leq_m \qDecide{Q^+}{\dltplus}$.
\end{lemma}

\begin{proof}
	Given an input hypergraph $\calG$ for the $\tetra{k}$ problem, we define an instance $I$ of $\qDecide{Q^+}{\dltplus}$.
	We consider a sequence of pseudo-minors $\calH(Q^+)=\calH_1,\calH_2,\ldots,\calH_m=\calH^b$, each pseudo-minor is obtained by performing one operation over the previous one, where $\calH_j=\calH^a$ for some $1\leq j\leq m$.
	We treat hypergraphs as describing CQs. That is, to the hypergraph $\calH_i$ we associate a query $Q_i$ such that $\calH(Q_i)=\calH_i$. Every edge $e$ of $\calH_i$ corresponds to an atom in $Q_i$ with a relational symbol $R^i_e$, and the vertices of $e$ are its variables. We assume that the variables in every atom are sorted by some total order. In the following, we construct instances $I_i$ to these queries. 
	It is possible to define a instance $I_m$ such that deciding $Q_m(I_m)$ solves $\tetra{k}$~\cite[Lemma 20]{bb:thesis}. We describe how to inductively build an instance $I_1=I$ such that deciding $Q^+(I)$ solves the same problem.
	
	\paragraph{Constructing $I$.}
	We first define $I_m$. For every edge $e$ of $\calH_{m}$, we define a relation
	$R_e^{m}$ that contains all edges of $\calG$ that have the same size as $e$. 
	A tuple of $R_e^{m}$ consists of the vertices of such an edge sorted by some total order on the vertices of $\calG$.
	We now define $I_i$ given $I_{i+1}$. We distinguish three cases according to the type of pseudo-minor operation that leads from $\calH_i$ to $\calH_{i+1}$.
	\begin{itemize}
		\item {\em edge removal}:
		For every $e''\in \calH_{i+1}$, set $R_{e''}^i=R_{e''}^{i+1}$.
		Then, let $e$ be the edge removed, and let $e'$ be an edge containing it.
		Set $R_e^{i}$ to be a copy of $R_{e'}^{i+1}$ projected accordingly.
		\item {\em edge contraction}: Let $v$ be the vertex replaced by its neighbor $u$.
		For any edge $e\in\calH_{i}$ contracting to an edge $e'\in\calH_{i+1}$, set $R_{e}^{i}$ to be a copy of $R_{e'}^{i+1}$, and assign the attribute $v$ a copy of the value of $u$ in every tuple. Then, if $u\not\in e$, project $u$ out of $R_{e}^{i}$.
		For every other edge $e''\in \calH_{i}$, set	 $R_{e''}^i=R_{e''}^{i+1}$.
		\item {\em vertex removal}: Let $v$ be the vertex removed, and let 
		$e\in\calH_{i}$ be an edge containing $v$ resulting in an edge $e'\in \calH_{i+1}$.
		Expand $R_{e'}^{i+1}$ to $R_{e}^{i}$ by copying $R_{e'}^{i+1}$, and assign $v$ with a constant $\bot$ in every tuple. Next apply the following FD-correction steps on $v$:
		\begin{enumerate}
			\item In every tuple, concatenate to the value of $v$ the values of variables it implies. These variables are defined via fixpoint iteration
as
\begin{align*}
\text{ImpliedBy}(v)&=\{v\}\text{ and}\\
\text{ImpliedBy}(v)&=\{w\mid t\rightarrow w\in\dltplus,t\in\text{ImpliedBy}(v)\}\cup \text{ImpliedBy}(v).
\end{align*}

For each $w\in\text{ImpliedBy}(v)\setminus\{v\}$, if $R^i(\vec{u})$ is an atom such that $\vec{u}[k]=v$ and $\vec{u}[j]=w$, then in every tuple $t\in R^i$, replace $t[k]$ with $(t[k],t[j])$.
			\item After the value of $v$ is determined, concatenate the value of $v$ to the variables implying it. These variable are defined by  
			\begin{align*}
	\text{Implies}(v)&=\{v\}\text{ and}\\
\text{Implies}(v)&=\{u\mid u\rightarrow t\in\dltplus,t\in\text{Implies}(v)\}
\cup \text{Implies}(v).			
			\end{align*}
			For each variable $u\in\text{Implies}(v)\setminus\{v\}$, if $R^i(\vec{u})$ is an atom such that $\vec{u}[k]=v$ and $\vec{u}[j]=u$, then in every tuple $t\in R^i$, replace $t[j]$ with $(t[j],t[k])$.
		\end{enumerate}
		For every edge $e''\in\calH_i$ not containing $v$, set $R_{e''}^i=R_{e''}^{i+1}$.
	\end{itemize}
	The overall construction of the instance $I$ can be done in linear time, since there is a constant number of pseudo-minor operations, each requiring a linear number of computational steps.
	
	\paragraph{$I$ is an instance of $\calS$. }
	We show that $I$ satisfies the FDs in $\dltplus$ by induction: we claim that for each $\calH_i$ all FDs $x\rightarrow y$ such that $x,y\in V(\calH_i)$ are satisfied.
	According to Lemma~\ref{lemma:fd-pm}, $\calH^a$ and therefore all of $\calH_j,\ldots,\calH_m$ do not contain all variables of any FD. Therefore our claim trivially holds for $\calH_j,\ldots,\calH_m$.
	We now prove our claim for $\calH_i$ where $i\leq j-1$. Consider an FD $\delta = x \rightarrow y$ such that $x,y\in V(\calH_i)$. There are three cases:
	\begin{itemize}
		\item If $x,y\in V(\calH_{i+1})$, then by the induction assumption $\delta$ is satisfied in $\calH_{i+1}$. If $\calH_{i+1}$ is obtained by edge removal, then the only new relation in $\calH_i$ is a projection of a relation of $\calH_{i+1}$, and therefore $\delta$ is satisfied in all relations. Otherwise, $\calH_{i+1}$ is obtained by vertex removal.
		If the value of $y$ is the same in $R_e^{i}$ as in $R_e^{i+1}$, we are done by the induction assumption. Otherwise, $y$ is changed due to the second FD-correction step, and the vertex removed is some $z$ such that $y \rightarrow z$.
		In this case, since $x$ transitively implies $z$, both $x$ and $y$ are concatenated with the same values, and $\delta$ is still satisfied.
		\item If $x\not\in V(\calH_{i+1})$, then $\calH_{i+1}$ is obtained by the removal of the vertex $x$, and the first FD-correction step ensures that $x$ contains a copy of the values of $y$ in every tuple where they both appear. Therefore $\delta$ is satisfied. 
		\item If $y\not\in V(\calH_{i+1})$, then $\calH_{i+1}$ is obtained by the removal of the vertex $y$. The second FD-correction step ensures that $x$ contains a copy of the values of $y$, and $\delta$ is satisfied.
	\end{itemize}
	
	\paragraph{Correctness. }
	We know~\cite[Lemma 20]{bb:thesis} that there is a solution to $Q_m(I_m)$ iff there exists a subhypergraph of $\calG$ isomorphic to $\calH_{t}$, and in fact every mapping $\mu$ that can be used for the evaluation corresponds to such a subhypergraph.
	We claim that every mapping used for evaluating $Q_{i+1}(I_{i+1})$ corresponds to a mapping that can be used for $Q_{i}(I_{i})$, and vice versa.
	This was already shown in case $\calH_{i+1}$ is obtained by $\calH_{i}$ via edge contraction \cite[Lemma 15]{bb:thesis} or edge removal \cite[Lemma 14]{bb:thesis},
	and it was shown for vertex removal~\cite[Lemma 13]{bb:thesis} if we simply assign the new vertex with a constant and skip the FD-correction steps.
	Let $\calH_{i+1}$ be a pseudo-minor obtained from $\calH_i$ via vertex removal, and denote by $I_i^0$ the instance constructed from $I_{i+1}$ as described but without the FD-correction steps. It is left to show that a mapping $\mu^0$ that satisfies $Q_i(I_i^0)$ corresponds to a mapping $\mu$ that satisfies $Q_i(I_i)$, and vice versa. This will conclude that $\calG$ has a subhypergraph isomorphic to $\calH_{t}$ iff $Q^+(I)\neq \emptyset$.
	
First we claim that if $x$ implies $y$ and \change{both $x$ and $y$ are present in $\calH_i$ (that is, for all $1\leq j < i$ the operation between $\calH_j$ and $\calH_{j+1}$ is not the removal of $x$ or $y$)}, then $y$ appears in every atom containing $x$ in $\calH_i$.
	This will help us show that we perform the same changes over $x$ in all relations. Since $Q^+$ is an FD-extension and only unary FDs are present, we are guaranteed that if $x$ implies $y$, then $y$ is present in every edge of $\calH(Q^+)=\calH_1$ where $x$ appears. This property is preserved under vertex removal and edge removal operations (as long as $x$ and $y$ are not removed), which are
	the only operations that can be performed between $\calH_1$ and $\calH_i$
	\change{(this is true since we perform vertex removal on $\calH_i$ and, by Definition~\ref{def:tetpm}, vertex-removal operations only occur between $\calH_1$ and $\calH^a$, and edge-removal is the only other operation allowed there).}
	
	We now show that given $\mu^0$ that satisfies $Q_i(I_i^0)$ there is a mapping $\mu$ that satisfies $Q_i(I_i)$, and vice versa. We show this by induction, considering one FD-correction step involving one variable at a time.
	Consider an FD-correction step of the first type on a vertex $v$ implying $w$, \change{where both $v$ and $w$ are present in $\calH_i$}. 
	In any atom $R(\vec{u})$ such that $\vec{u}[k]=v$, we showed that there exists an index $j$ such that $\vec{u}[j]=w$. For every tuple $t_0\in R^0$, there is a similar tuple $t\in R$ with the only difference being $t[k]=(t_0[k],t_0[j])$.
	Therefore, by defining $\mu(v)=(\mu_0(v),\mu_0(w))$ and $\mu(u)=\mu_0(u)$ for all other variables $u\neq v$, every tuple that is used in the evaluation of $\mu^0$ in $I_i^0$ results in a tuple that can be used in the evaluation of $\mu$ in $I_i$. Indeed, $\mu$ is a valid evaluation of $Q_i(I_i)$.
	A similar argument holds similarly for the opposite direction and for the second type of FD-correction step.
	For the opposite direction for example, if $\mu(v)=(a_v,a_w)$, we define $\mu^0(v)=a_v$.
\end{proof}

Theorem~\ref{thm:cyclicHardness} is an immediate consequence of Lemma~\ref{lemma:simple-pm}:

\begin{proof}[Proof of Theorem~\ref{thm:cyclicHardness}] For the sake of a contradiction assume that $Q$ is FD-cyclic, and $\qDecide{Q}{\Delta}$ is solvable in linear time. Theorem~\ref{thm:qpEquiv} implies a linear time reduction $\qDecide{Q^+}{\dltplus}\leq_m\qDecide{Q}{\Delta}$. Therefore, it is possible to solve $\qDecide{Q^+}{\dltplus}$ in linear time as well. As $Q^+$ is cyclic, there exists a pseudo-minor $\calH_{pm}\in\tetpm{Q^+}$ isomorphic to $\tetra{k}$ for some $k\geq3$. According to Lemma~\ref{lemma:simple-pm}, this $\tetra{k}$ problem is also solvable in linear time.
\end{proof}

In terms of enumeration complexity, Theorem~\ref{thm:cyclicHardness} means that any enumeration algorithm for such a query cannot output a first solution (or decide that there is none) within linear time, and we get the following corollary. 

\begin{corollary}\label{cor:cyclicHardness}
	Let $Q$ be a self-join free CQ over a schema $\calS=(\calR,\Delta)$, where $\Delta$ only contains unary FDs. If $Q$ is FD-cyclic, then $\qEnum{Q}{\Delta}\not\in\DelayClin$,
	assuming that the $\tetra{k}$ problem is not solvable in linear time for any $k$.
\end{corollary}

Less restrictive than constant delay enumeration, the class $\DelayLin$ consists of enumeration problems that can be solved with a linear delay between solutions.
Acyclic CQs are known to be in $\DelayLin$~\cite{bdg:dichotomy}, and Corollary~\ref{cor:posExtension} shows that FD-acyclic CQs are in this class as well.
Theorem~\ref{thm:cyclicHardness} implies a lower bound for $\DelayLin$ as well.
Thus, we obtain a dichotomy stating that CQs are in $\DelayLin$ iff they are FD-acyclic.

\begin{theorem}\label{theorem:lineardelaydichomtomy}
	Let $Q$ be a CQ with no self-joins over a schema $\calS=(\calR,\Delta)$, where $\Delta$ only contains unary FDs.
	\begin{itemize}
		\item If $Q$ is FD-acyclic, then $\qEnum{Q}{\Delta}\in\DelayLin$.
		\item Otherwise (if $Q$ is FD-cyclic), $\qEnum{Q}{\Delta}\not\in\DelayLin$, 
		assuming that the $\tetra{k}$ problem cannot be solved in linear time for any $k$.
	\end{itemize}
\end{theorem}

\begin{proof}
	If $Q^+$ is acyclic, then $\qEnum{Q^+}{\emptyset}\in\DelayLin$~\cite{bdg:dichotomy}.
	According to Theorem~\ref{thm:qpEquiv}, $\qEnum{Q}{\Delta}\leq_e\qEnum{Q^+}{\dltplus}$.
	Since every instance that satisfies $\dltplus$ also satisfies $\emptyset$, we conclude that $\qEnum{Q^+}{\dltplus}\leq_e\qEnum{Q^+}{\emptyset}$ using the identity mapping.
	Therefore, $\qEnum{Q}{\Delta}\in\DelayLin$, since $\DelayLin$ is closed under exact reductions.
	
	If $Q^+$ is cyclic, assume by contradiction that $\qEnum{Q}{\Delta}\in\DelayLin$. According to Corollary~\ref{cor:cyclicHardness}, $\qEnum{Q^+}{\dltplus}\in\DelayLin$ as well. This means that finding a first answer to $Q^+$ or deciding that there is none can be done in linear time, in contradiction to Theorem~\ref{thm:cyclicHardness}.
\end{proof}

We conclude this section with a short discussion about its extension to general FDs. The following example shows that the proof for Theorem~\ref{thm:cyclicHardness} that was provided here cannot be lifted
to general FDs.
Exploring this extension is left for future work.

\begin{figure}[t]
	\begin{center}
		\begin{tikzpicture}[ dot/.style={draw, circle, inner sep=1pt, fill=black}]
		
		\node[dot,label=above:\(y\)] (v1) at (0,0) {};
		\node[dot,label=above:\(v\)] (v2) at (3,0) {};
		\node[dot,label=above:\(z\)] (v3) at (6,0) {};
		\node[dot,label=above:\(w\)] (v4) at (4.5,1.5) {};
		\node[dot,label=above:\(x\)] (v5) at (3,3) {};
		\node[dot,label=above:\(u\)] (v6) at (1.5,1.5) {};
		
		\draw \convexpath{v1,v2,v3}{0.55cm};
		\draw \convexpath{v1,v6,v5}{0.55cm};
		\draw \convexpath{v5,v3,v4}{0.55cm};
		\draw \convexpath{v6,v4,v2}{0.55cm};
		
		\end{tikzpicture}
	\end{center}
	\caption{The hypergraph $\calH(Q^+)$ for Example~\ref{exp:generalCyclic}}
	\label{fig:cookie}
\end{figure}

\begin{example}\label{exp:generalCyclic}
	
	Consider $Q()\leftarrow R_1(x,y,u),R_2(x,w,z),R_3(y,v,z),R_4(u,v,w)$ over a schema with all possible two-to-one FDs in $R_1$, $R_2$ and $R_3$. That is,
\begin{align*}
 \Delta =\{& xy\rightarrow u, yu\rightarrow x, ux\rightarrow y,
	zy\rightarrow v,\\
	&yv\rightarrow z,
	 vz\rightarrow y,
	xz\rightarrow w, zw\rightarrow x, wx\rightarrow z\}.
\end{align*}	
	Note that $Q^+=Q$. 
	The hypergraph $\calH(Q^+)$ is cyclic, see Figure~\ref{fig:cookie}, yet it is unclear whether $Q$ can be solved in linear time, and whether $\tetra{3}$ can be reduced to answering $Q^+$.
	Using Lemma~\ref{lemma:fd-pm}, $\calH(Q^+)$ has triangle pseudo-minors that do not contain all variables of any FD. Consider for example the one obtained by removing all vertices other than $x,y,z$. A construction similar to that of Lemma~\ref{lemma:simple-pm} would assign $u$ with the values of $x$ and $y$, assign $v$ with the values of $y$ and $z$, and assign $w$ with the values of $x$ and $z$. This results in the edge $\{u,v,w\}$ containing all three values of any possible triangle, meaning that this edge cannot be constructed in linear time.
	Other choices of triangle pseudo-minors lead to similar encoding problems. 
	\qed
\end{example}

\section{Cardinality Dependencies}\label{section:cardinalitydep}

In this section, we show that our results also apply to CQs over schemas with cardinality dependencies.
{\em Cardinality Dependencies} (CDs)~\cite{DBLP:conf/icalp/ArapinisFG16,DBLP:journals/pvldb/CaoFWY14} are a generalization of FDs, where the left-hand side does not uniquely determine the right-hand side,
but rather provides a bound on the number of distinct values it can have.
Formally, $\Delta$ is the set of $CDs$ of a schema $\calS=(\calR,\Delta)$. Every $\delta\in\Delta$ has the form $(R_i\colon A\rightarrow B,c)$, where $R_i\colon A\rightarrow B$ is an FD and $c$ is a positive integer.
A CD $\delta$ is {\em satisfied} by an
instance $I$ over $\calS$, if 
every set of tuples $S\subseteq(R_i)^I$ that agrees on the indices of $A$, but no pair of them agrees on all indices of $B$, contains at most $c$ tuples.
It follows from the definition that $\delta$ is an FD if $c=1$.

Denote by $\Delta^\text{FD}$ the FDs obtained from a set of CDs $\Delta$ by setting all $c$ values to one.
Given a query $Q$ over $\calS=(\calR,\Delta)$, we define the {\em CD-extended query} $Q^+$ of $Q$ to be the FD-extended query of $Q$ over $\calS=(\calR,\Delta^\text{FD})$.
The schema $\calS^+$ is defined with the original $c$ values, and the extended CDs are $\dltplus=\{(R_i^+\colon A\rightarrow b,c)\mid \exists(R_j\colon A\rightarrow B,c)\in \Delta, b\in B, A\cup\{b\}\subseteq\var(R_i^+)\}$.
Note that FD-extensions are indeed a special case of CD-extensions.

The hardness results extend to CDs because they are not more restrictive than FDs: 
Since every instance that preserves the FDs $\Delta^\text{FD}$ also preserves the CDs $\Delta$, we conclude that $\qEnum{Q}{\Delta^{\text{FD}}}\leq_e\qEnum{Q}{\Delta}$.
When only FDs are present we can apply Theorem~\ref{thm:qpEquiv}, and get $\qEnum{Q^+}{\dltplus^{\text{FD}}}\leq_e\qEnum{Q}{\Delta^{\text{FD}}}$. Combining the two we get the following lemma.

\begin{lemma}\label{lemma:negCD}
	Let $Q$ be a \change{self-join free} CQ over a schema $\calS=(\calR,\Delta)$, where $\Delta$ is a set of CDs,
	and let $Q^+$ be its CD-extension. Then
	$\qEnum{Q^+}{\dltplus^{\text{FD}}}\leq_e\qEnum{Q}{\Delta}$.
\end{lemma}

Defining the classes of {\em CD-acyclic} and {\em CD-free-connex} queries similarly to the case with FDs, Lemma~\ref{lemma:negCD} implies that all lower bounds presented in this paper hold for CDs.
If $Q$ is self-join free and CD-acyclic but not CD-free-connex and $\qEnum{Q}{\Delta}\in\DelayClin$, then by Lemma~\ref{lemma:negCD} we have that $\qEnum{Q^+}{\dltplus^\text{FD}}\in\DelayClin$ as well. According to Lemma~\ref{lemma:acyclicgeneralfdnegative} this means that $\qEnum{\Pi}{\emptyset}\in\DelayClin$, and the matrix multiplication problem can be solved in quadratic time. So $\qEnum{Q}{\Delta}\not\in\DelayClin$, assuming the boolean matrix multiplication assumption.
Similarly, we conclude the hardness of self-join free CD-cyclic CQs over schemas that contain only unary CDs, of the form $(A\rightarrow B,c)$ with $|A|=1$.  Combining Lemma~\ref{lemma:negCD} with Theorem~\ref{thm:cyclicHardness}, we have that such queries cannot be evaluated in linear time, assuming that the $\tetra{k}$ problem cannot be solved in linear time for any $k$.

In order to extend the positive results,
we need to show that the CD-extension is at least as hard as the original query w.r.t.~enumeration. We use a slight relaxation of exact reductions: For $\pEnum{R_1}\leq_{e'}\pEnum{R_2}$, instead of a bijection between the sets of outputs, one output of $\pEnum{R_1}$ corresponds to at most a constant number of outputs of $\pEnum{R_2}$.

\begin{lemma}\label{lemma:posCD}
	Let $Q$ be a \change{self-join free} CQ over a schema $\calS=(\calR,\Delta)$, where $\Delta$ is a set of CDs,
	and let $Q^+$ be its CD-extension. Then
	$\qEnum{Q}{\Delta}\leq_{e'}\qEnum{Q^+}{\dltplus}$.
\end{lemma}

\begin{proof}
	When dealing with FDs, we assume that the right-hand side has only one variable, as we can use such FDs to describe all possible ones. With CDs this no longer holds. Nonetheless, every instance of the schema $\calS=(\calR,\Delta)$ satisfies $\Delta^1=\{(R_i\colon A\rightarrow b,c)\mid (R_i\colon A\rightarrow B,c)\in\Delta,b\in{B}\}$, so is also an instance of $\calS^1=(\calR,\Delta^1)$. Therefore, $\qEnum{Q}{\Delta}\leq_e\qEnum{Q}{\Delta^1}$ using the identity mapping.
	It is left to show that $\qEnum{Q}{\Delta^1}\leq_{e'}\qEnum{Q^+}{\dltplus}$.
	The proof idea is the same as in Theorem~\ref{thm:qpEquiv}, except now, for each tuple extended from $R^I_i$ to $R^{I^+}_i$ we can have at most $c$ new tuples. Since this process is only done a constant number of times, the construction still only requires linear time, and the rest of the proof holds.
	Note that now one solution of $\qEnum{Q^+}{\Delta^+}$ may correspond to several solutions of $\qEnum{Q}{\Delta^1}$, as some variables were possibly added to the head. However, as the possible values of the added head variables are bounded by CDs, the number of solutions of $Q^+$ that correspond to one solution of $Q$ is bounded by a constant.
	
	We now formally prove that $\qEnum{Q}{\Delta^1}\leq_{e'}\qEnum{Q^+}{\dltplus}$.
	Denote $Q$ by $Q(\vec{p}) \leftarrow R_1(\vec{v}_1), \dots, R_m(\vec{v}_m)$.
	Given an instance $I$ of $\qEnum{Q}{\Delta}$, we define $\sigma(I)$.
	We start by removing tuples that interfere with the extended dependencies. For every dependency $\delta=(R_j\colon X\rightarrow y,c)\in\Delta^1$ and every atom $R_k(\vec{v}_k)$ that contains $X\cup\{y\}$, we correct $R_k$ according to $\delta$:
	we only keep tuples of $R_k^I$ that agree with some tuple of $R_j^I$ over the values of
	$X\cup\{y\}$.
	Next, we follow the extension of the schema and extend the instance accordingly. This phase results in a sequence of instances $I_0,I_1,\ldots,I_n=\sigma(I)$ that correspond to a sequence of queries $Q=Q_0,Q_1,\ldots,Q_n=Q^+$ such that each query is the result of extending an atom or the head of the previous query according to an FD. If in step $i$ the head was extended, we set $I_{i+1}=I_i$. Now assume some relation $R_k$ is extended according to some  CD $(R_j\colon X\rightarrow y,c)$.
	For each tuple $t\in R^{I_i}_k$, if there is no tuple $s\in R^{I_i}_j$ that agrees with $t$ over the values of $X$, then we remove $t$ altogether.
	Otherwise, we consider all values such tuples assign $y$. 
	Denote those values by $a_1,\ldots,a_m$, and note that due to the CD, $m\leq c$.
	We copy $t$ to $R^{I_{i+1}}_k$ $m$ times, each time assigning $y$ with a different value of $a_1,\ldots,a_m$. Given an answer $\mu\in Q^+(\sigma(I))$, we define $\tau(\mu)$ to be the projection of $\mu$ to $\free(Q)$.

	We need to show that $Q(I)=\{\mu|_{\free(Q)}:\mu|_{\free(Q^+)}\in Q^+(\sigma(I))\}$, and that an element of the left-hand side may only appear a constant amount of times on the right-hand side.
	First, if $\mu|_{\free(Q^+)}\in Q^+(\sigma(I))$, since all tuples of $\sigma(I)$ appear (perhaps projected) in $I$, then $\mu|_{\free(Q)}\in Q(I)$.
It is left to show the opposite direction: if $\mu|_{\free(Q)}\in Q(I)$ then $\mu|_{\free(Q^+)}\in Q^+(I^+)$.
We show by induction on $Q=Q_0,Q_1,\ldots,Q_n=Q^+$ that $\mu|_{\free(Q_i)}\in Q_i(I_i)$.
The induction base holds since in the cleaning phase we did not remove ``useful'' tuples.
Since $\mu|_{\free(Q)}\in Q(I)$,
there exist tuples, one of each relation of the query, that agree on the values of $X\cup\{y\}$ (they all assign them with the values $\mu$ assigns them). These tuples were not removed in the cleaning phase, and therefore $\mu|_{\free(Q)}\in Q(I_0)$.
Next assume that $\mu|_{\free(Q_i)}\in Q_i(I_i)$, and we want to show that $\mu|_{\free(Q_{i+1})}\in Q_{i+1}(I_{i+1})$. This claim is trivial in case the head was extended.
Note that there can be at most $c-1$ different answers $\mu'|_{\free(Q_{i+1})}$  in $Q_{i+1}(I_{i+1})$ such that $\mu|_{\free(Q_{i+1})}\neq\mu'|_{\free(Q_{i+1})}$ but $\mu|_{\free(Q_{i})}=\mu'|_{\free(Q_{i})}$, as the added variable $y$ is bound by the CD to have at most $c$ possible values.
Now consider the case where an atom $R_k(\vec{v}_k)$ was extended according to a CD $(R_j\colon X\rightarrow y,c)$.
The tuple $\mu(\vec{v}_k)\in R_k^{I_i}$ was extended with the value $\mu(y)$ due to the tuple $\mu(\vec{v}_j)\in R_j^{I_i}$ that agrees with it on the values of $X$, and so $\mu(\vec{v}_k,y)\in R_k^{I_{i+1}}$. In case of self-joins, other atoms with the relation $R_k$ are extended with a new and distinct variable. Such variables will be mapped to this value $\mu(y)$ as well. Overall, we have that $\mu$ (extended by mappings of the fresh variables) is also a homomorphism in $Q_{i+1}(I_{i+1})$.	
\end{proof}

We can now extend our positive results to accommodate CDs. Let $Q$ be a CD-free-connex CQ over a schema $\calS=(\calR,\Delta)$, where $\Delta$ contains CDs.
\change{
We have $\qEnum{Q}{\Delta}\leq_{e}\qEnum{\SJF(Q)}{\SJF(\Delta)}\leq_{e'}\qEnum{
\SJF(Q)^+}{\SJF(\Delta)_{Q^+}}\leq_e\qEnum{\SJF(Q)^+}{\emptyset}$ according to Lemma~\ref{lemma:posCD},
and $\qEnum{\SJF(Q)^+}{\emptyset}\in\DelayClin$ due to Theorem~\ref{theorem:originalDichotomy}.}
The class $\DelayClin$ is closed under this type of reduction. To avoid printing duplicates, we need to store previous results in a lookup table, and verify that a generated result is new before printing it.
This alone is not enough, as we can have a long sequence of generating known results, and then the delay between generating new results can be larger than constant. For this reason, we use the known technique of delaying the results,
also used by Capelli et. al~\cite[Proposition 12]{CAPELLI2018}.
If every answer to $\qEnum{Q}{\Delta}$ corresponds to at most $c$ answers to $\qEnum{Q^+}{\emptyset}$, we save the newly generated results in a queue, and after generating $c$ results we pop and print a result from the queue. This guarantees that the queue is never empty when accessed, and the results are printed with constant delay. Therefore, $\qEnum{Q}{\Delta}\in\DelayClin$, and we deduce the following:

\begin{theorem}\label{thm:dichotomyCDs}
Let $Q$ be a CD-acyclic CQ over the schema $\calS=(\calR,\Delta)$, where $\Delta$ is a set of CDs.
\begin{itemize}
\item If $Q$ is CD-free-connex, then $\qEnum{Q}{\Delta}\in\DelayClin$.
\item If $Q$ is \change{self-join free} and not CD-free-connex, then $\qEnum{Q}{\Delta}\not\in\DelayClin$,  assuming
 that the product of two $n \times n$ boolean matrices cannot be computed in time $\calO(n^2)$.
\end{itemize}
\end{theorem}

\section{Conjunctive Queries with Disequalities}\label{sec:disequalities}

In this last section, we show that all results in this article apply to CQs with disequalities.
To keep the proof ideas clear, we begin by assuming that the schema only contains FDs. We incorporate our results with CDs at the end of this section.
A {\em CQ with disequalities} over a schema $\calS=(\calR,\Delta)$ is an expression of the form
$$Q(\vec{p}) \leftarrow R_1(\vec{v}_1), \dots, R_m(\vec{v}_m),w_1\neq w_2, \ldots , w_{k-1}\neq w_k$$ 
where $Q(\vec{p}) \leftarrow R_1(\vec{v}_1), \dots, R_m(\vec{v}_m)$ is a CQ denoted by $\rel{Q}$, and $w_1,\ldots,w_k$ are variables in $\var(Q)$.
We denote $\dis{Q}=\{w_1\neq w_2, \ldots , w_{k-1}\neq w_k\}$.
The evaluation is 
$Q(I)=\{\mu |_{\vec{p}}\in \rel{Q}(I) | \forall w_i\neq w_j \in\dis{Q}: \mu(w_i)\neq \mu(w_j)\}$.

The structural definitions of CQs with disequalities depend only on their bases. That is, $Q$ is said to be {\em acyclic} if $\rel{Q}$ is acyclic, and we similarly define {\em acyclic} and {\em free-connex} CQs with disequalities. 
The FD-extension of a CQ with disequalities is denoted by $Q^+$. The base extends as before, and the disequalities remain the same. That is, $\rel{(Q^+)}=\rel{Q^+}$ is the FD-extension of ${\rel{Q}}$, and $\dis{Q^+}=\dis{Q}$. We say that $Q$ is {\em FD-acyclic} if $\rel{Q}$ is FD-acyclic, and similarly define {\em FD-free-connex} and {\em FD-cyclic} CQs with disequalities.

We first show that our positive results apply also to CQs with disequalities. The following lemma is the equivalent of 
the positive case of Theorem~\ref{thm:qpEquiv}. In the proof, the reduction from the query to its extension works as before for the base query, and the disequalities are satisfied as they remain the same during the extension.

\begin{lemma}\label{lemma:disPos}
	Let $Q$ be a \change{self-join free} CQ with disequalities over a schema $\calS=(\calR,\Delta)$, and let $Q^+$ be its FD-extension.
	Then $\qEnum{Q}{\Delta}\leq_e\qEnum{Q^+}{\dltplus}$.
\end{lemma}

\begin{proof}
	Recall that $\rel{Q^+}$ is the FD-extension of $\rel{Q}$.
	According to the proof of Theorem~\ref{thm:qpEquiv}, we can show that $\qEnum{\rel{Q}}{\Delta}\leq_e\qEnum{\rel{Q^+}}{\dltplus}$ by using a construction $\sigma$ such that $\mu|_{\free(Q)}\in \rel{Q}(I)$ iff $\mu|_{\free(Q^+)}\in \rel{Q}^+(\sigma(I))$.
	We use this same construction to show that $\qEnum{Q}{\Delta}\leq_e\qEnum{Q^+}{\dltplus}$.
	By definition we have that $\mu|_{\free(Q)}\in Q(I)$ iff $\mu|_{\free(Q)}\in \rel{Q}(I)$ and in addition $\forall u\neq w \in \dis{Q}: \mu(u)\neq\mu(w)$.
	As in the proof of Theorem~\ref{thm:qpEquiv}, and since $\dis{Q} = \dis{Q^+}$ , this is true iff $\mu|_{\free(Q^+)}\in \rel{Q}^+(\sigma(I))$ and $\forall u\neq w \in \dis{Q^+}: \mu(u)\neq\mu(w)$.
	This is equivalent to $\mu|_{\free(Q^+)}\in Q^+(\sigma(I))$.
	We conclude that $\mu|_{\free(Q)}\in Q(I)$ iff $\mu|_{\free(Q^+)}\in Q^+(\sigma(I))$.
\end{proof}

Combining Lemma~\ref{lemma:disPos} with the full version of Theorem~\ref{theorem:originalDichotomy} that holds for CQs with disequalities~\cite{bdg:dichotomy}, we have that if $Q$ is an FD-free-connex CQ with disequalities over a schema $\calS=(\calR,\Delta)$, then $\qEnum{Q}{\Delta}\in\DelayClin$.

We now discuss the lower bound. We formalize an idea proposed by Bagan et al. to extend their negative results to CQs with disequalities, and show that this idea can be applied to reductions in general, and to our case in particular. 

\begin{lemma}\label{lemma:disHardReduction}
	Let $Q$ be a CQ with disequalities over a schema $\calS=(\calR, \Delta)$.
	Then $\qEnum{\rel{Q}}{\Delta}\leq_e\qEnum{Q}{\Delta}$.
\end{lemma}

\begin{proof}
	Given an instance $I$ of $\qEnum{\rel{Q}}{\Delta}$, we construct $\sigma(I)$ by assigning every variable a disjoint domain. Formally, for every atom $R(\vec{v})$ of $\rel{Q}$ and every tuple $(c_1,\ldots,c_t)\in R^I$, we have the tuple $((c_1,\vec{v}[1]),\ldots,(c_t,\vec{v}[t]))$ in $R^{\sigma(I)}$. 
	We claim that the answers of $\rel{Q}$ over $I$ are exactly those of $Q$ over $\sigma(I)$ if we omit the variable names.
	That is, we define $\tau : \idom\times\var(Q) \rightarrow \idom$ as $\tau((c,v))=c$, and show $\rel{Q}(I)=\tau(Q(\sigma(I)))$.
	
	Let $\mu |_{\free(Q)}\in \rel{Q}(I)$. Then we know that for every atom $R(\vec{v})$ of $\rel{Q}$, there exists a tuple 	$(\mu(\vec{v}[1]),\ldots,\mu(\vec{v}[t]))\in R_I$. By definition of $\sigma$, it results in $((\mu(\vec{v}[1]),\vec{v}[1]),\ldots,(\mu(\vec{v}[t]),\vec{v}[t]))\in R^{\sigma(I)}$.
	Therefore, as we define the mapping $f_\mu: \var(Q) \rightarrow\idom\times\var(Q)$ to be $f_\mu(u)=(\mu(u),u)$, we have that $f_\mu|_{\free(Q)}\in \rel{Q}(\sigma(I))$. 
	Note that for all $u,w\in\var(Q)$ with $u\neq w$ we have $f_\mu(u)\neq f_\mu(w)$. Therefore, $f_\mu\in Q(\sigma(I))$ as all disequalities in $\dis{Q}$ are satisfied. Since $\tau\circ f_\mu = \mu$, we have that $\mu |_{\free(Q)}\in \tau(Q(\sigma(I)))$, and this concludes that $\rel{Q}(I)\subseteq\tau(Q(\sigma(I)))$.
	The opposite direction holds as well.
	If $\nu |_{\free(Q)}\in Q(\sigma(I))$, then for every atom $R(\vec{v})$ in $Q$ we have that $\nu(\vec{v})\in R^{\sigma(I)}$. By construction, $\tau(\nu(\vec{v}))\in R^{I}$, and therefore $\tau\circ\nu|_{\free(Q)}\in \rel{Q}(I)$.
\end{proof}

By combining Theorem~\ref{thm:qpEquiv} with Lemma~\ref{lemma:disHardReduction}, we get the following reduction.

\begin{lemma}\label{lemma:disNeg}
	Let $Q$ be a \change{self-join free} CQ with disequalities over a schema $\calS=(\calR, \Delta)$.
	Then $\qEnum{\rel{Q}^+}{\dltplus}\leq_e\qEnum{Q}{\Delta}$.
\end{lemma}

Lemma~\ref{lemma:disNeg} means that the hardness results of this paper extend to CQs with disequalities. In particular, if a self-join free $Q$ is FD-acyclic but not FD-free-connex, then $\qEnum{\Pi}{\emptyset}\leq_e\qEnum{\rel{Q}^+}{\dltplus}\leq_e \qEnum{Q}{\Delta}$ by Lemma~\ref{lemma:acyclicgeneralfdnegative} and Lemma~\ref{lemma:disNeg}. We get that then $\qEnum{Q}{\Delta}\not\in\DelayClin$, assuming that the product of two $n \times n$ boolean matrices cannot be computed in time $\calO(n^2)$.
Similarly, we can use Lemma~\ref{lemma:simple-pm} and Lemma~\ref{lemma:disNeg} to show that if a self-join free $Q$ is FD-cyclic, and the schema only contains unary FDs, then $\qDecide{Q}{\Delta}$ cannot be solved in linear time, assuming that the $\tetra{k}$ problem is not solvable in linear time for any $k$.

We next consider CQs with disequalities in the presence of CDs.
The ideas presented in Section~\ref{section:cardinalitydep} and Section~\ref{sec:disequalities} can be combined to show the following lemma.

\begin{lemma}\label{lemma:extensions}
	Let $Q$ be a \change{self-join free} CQ with disequalities over a schema $\calS=(\calR,\Delta)$, where $\Delta$ is a set of CDs.
	 We have:
	\begin{itemize}
		\item $\qEnum{Q}{\Delta}\leq_{e'}\qEnum{Q^+}{\dltplus}$
		\item $\qEnum{\rel{Q}^+}{\dltplus^{\text{FD}}}\leq_e\qEnum{Q}{\Delta}$
	\end{itemize}
\end{lemma}

The proof for $\qEnum{Q}{\Delta}\leq_{e'}\qEnum{Q^+}{\dltplus}$ works similarly to that of Lemma~\ref{lemma:disPos}, except it builds upon Lemma~\ref{lemma:posCD} instead of Theorem~\ref{thm:qpEquiv}. We use the same reduction as in Lemma~\ref{lemma:posCD}, and the disequalities remain unchanged.
By combining Lemma~\ref{lemma:disHardReduction} with Lemma~\ref{lemma:negCD}, he have that by transitivity $\qEnum{\rel{Q}^+}{\dltplus^{\text{FD}}}\leq_e\qEnum{\rel{Q}}{\Delta}\leq_e\qEnum{Q}{\Delta}$.

Lemma~\ref{lemma:extensions} along with Lemma~\ref{lemma:acyclicgeneralfdnegative}, Lemma~\ref{lemma:simple-pm} and the positive results by Bagan et al.~\cite[Theorem 13, Theorem 17]{bdg:dichotomy} prove that all results presented in this article apply to CQs with disequalities over schemas with cardinality dependencies.
\change{To prove the positive results, we need to go as before through the self-join free version of the query: $\qEnum{Q}{\Delta}\leq_{e}\qEnum{\SJF(Q)}{\SJF(\Delta)}\leq_{e'}\qEnum{
\SJF(Q)^+}{\SJF(\Delta)_{Q^+}}\leq_e\qEnum{\SJF(Q)^+}{\emptyset}$.}
The following theorem summarizes all classification results presented here, as FDs are a special case of CDs, and CQs are a special case of CQs with disequalities.
\begin{theorem}
	Let $Q$ be a CQ with disequalities over a schema $\calS=(\calR,\Delta)$, where $\Delta$ is a set of CDs.
	\begin{itemize}
		\item If $Q$ is CD-free-connex, then $\qEnum{Q}{\Delta}\in\DelayClin$.
		\item If $Q$ is CD-acyclic, then $\qEnum{Q}{\Delta}\in\DelayLin$.
		\item If $Q$ is self-join free, CD-acyclic but not CD-free-connex, then $\qEnum{Q}{\Delta}\not\in\DelayClin$, assuming
		that the product of two $n \times n$ boolean matrices cannot be computed in time $\calO(n^2)$.
		\item If $Q$ is self-join free, CD-cyclic and $\Delta$ contains only unary CDs, then $\qDecide{Q}{\Delta}$ cannot be solved in linear time, assuming that the $\tetra{k}$ problem is not solvable in linear time for any $k$. In particular, $\qEnum{Q}{\Delta}\not\in\DelayClin$ and $\qEnum{Q}{\Delta}\not\in\DelayLin$.
	\end{itemize}
\end{theorem}

\section{Concluding Remarks}\label{section:conclusion}

Previous hardness results regarding the enumeration complexity of CQs no longer hold in the presence of dependencies.
In this paper, we have shown that some of the queries which where previously classified as hard are in fact tractable in the presence of FDs, and that the others remain intractable.
We have classified the enumeration complexity of self-join free CQs according to their FD-extension. Under previously used complexity assumptions: a query is in $\DelayClin$ if its extension is free-connex, it is not in $\DelayClin$ if its extension is acyclic but not free-connex, and it is not even decidable in linear time if the schema has only unary FDs and its extension is cyclic.
We also show that these results apply for CQs with disequalities and in the presence of CDs.
In addition to our results on constant delay enumeration, we show that this work has consequences in other enumeration classes such as $\DelayLin$.

This work opens up quite a few directions for future work.
Our proof for the hardness of FD-cyclic CQs assumes that all FDs are unary. The question of whether this result holds for general FDs, along with the classification of Example~\ref{exp:generalCyclic}, remains open.
In addition, to show that enumerating CD-free-connex CQs can be done in $\DelayClin$, we store all printed results. The required space has the size of the output, which may be polynomial in that of the input
in data complexity and exponential in combined complexity.
It is unclear whether there exists a solution that requires less space.
Finally, we wish to explore how the tools provided here can be used to extend other known results on query enumeration, such as a dichotomy for enumerating CQs with negation~\cite{bb:thesis}, to accommodate FDs.

\bibliographystyle{spmpsci}      
\bibliography{references}

\end{document}